\pdfoutput=1
\documentclass[12pt]{article}

\usepackage[T1]{fontenc}
\usepackage[utf8]{inputenc}%
\usepackage[tracking=true, letterspace=50, expansion=false]{microtype}

\usepackage{amsmath, amsfonts, amssymb, amsthm, fullpage, color, bbm, enumerate, titlesec, graphicx, epstopdf, multirow, array, mathtools}
\usepackage{subcaption}
\usepackage{lmodern}
\usepackage[title]{appendix}
\usepackage[onehalfspacing]{setspace}

\usepackage[runin]{abstract}

\usepackage[font={small}]{caption}

\definecolor{darkred}{rgb}{0.5,0,0}

\titleformat*{\paragraph}{\sc}
\usepackage[longnamesfirst]{natbib}
\usepackage{hyperref}
\hypersetup{%
  pdftitle=Dynamic Causal Effects in a Nonlinear World,%
  pdfauthor=Michal Kolesár and Mikkel Plagborg-Møller,
  allbordercolors={1 1 1},
  allcolors=darkred,
  colorlinks=true}

\usepackage[nameinlink,noabbrev]{cleveref}

\newcolumntype{C}[1]{>{\centering\arraybackslash}p{#1}}

\theoremstyle{plain}

\crefname{asn}{Assumption}{Assumptions}
\newtheorem{lem}{Lemma}
\crefname{lem}{Lemma}{Lemmas}

\newtheorem{prop}{Proposition}
\crefname{prop}{Proposition}{Propositions}

\crefname{cor}{Corollary}{Corollaries}

\theoremstyle{definition}
\newtheorem{exm}{Example}
\AtBeginEnvironment{exm}{%
  \pushQED{\qed}%
}
\AtEndEnvironment{exm}{\popQED\endexm}

\newcommand\independent{\protect\mathpalette{\protect\independenT}{\perp}}
\def\independenT#1#2{\mathrel{\rlap{$#1#2$}\mkern2mu{#1#2}}}

\DeclarePairedDelimiter\abs{\lvert}{\rvert} 
\DeclarePairedDelimiter\indicatorfence{\{}{\}}
\newcommand\1{\operatorname{\mathbbm{1}}\indicatorfence}

\DeclareMathOperator*{\argmin}{argmin}

\DeclareMathOperator*{\var}{Var}
\DeclareMathOperator*{\cov}{Cov}

\usepackage{bm}
\renewcommand{\vec}[1]{\boldsymbol{\mathbf{#1}}}
\newcommand{\Ra}[0]{R}

\usepackage[nolist]{acronym}
\begin{acronym}
  \acro{VAR}{vector autoregression}
  \acro{LLP}{linear local projection onto observed  shocks or proxies}%
  \acroplural{LLP}[LLPs]{linear local projections onto observed shocks or proxies}%
  \acro{ATE}{average treatment effect}
  \acro{HJ}{Herbst and Johannsen}
  \acro{GHP}{Gonçalves, Herrera, and Pesavento}
  \acro{KWX}{Kitagawa, Wang, and Xu}
  \acroplural{VAR}[VARs]{vector autoregression}
  \acro{DGP}{data generating process}
  \acroplural{DGP}[DGPs]{data generating processes}
  \acro{OLS}{ordinary least squares}
  \acro{CDF}{cumulative distribution function}
  \acro{ICA}{independent components analysis}
\end{acronym}

\allowdisplaybreaks

\crefname{sappsec}{Supplemental Appendix}{Supplemental Appendices}
\crefname{sappsubsec}{Supplemental Appendix}{Supplemental Appendices}
\crefname{sappsubsubsec}{Supplemental Appendix}{Supplemental Appendices}
\crefname{appsec}{Appendix}{Appendices}

\usepackage{xr}
\begin{document}

\title{\texorpdfstring{\vspace{-1.5\baselineskip}}{} Dynamic Causal Effects in a
  Nonlinear World: \texorpdfstring{\\}{}the Good, the Bad, and the
  Ugly\thanks{Email: {\tt mkolesar@princeton.edu}, {\tt mikkelpm@princeton.edu}.
    We are grateful for comments from %
    Joshua Angrist, %
    Paul Bousquet, %
    Alessandro Casini, %
    Runtong Ding, %
    Florian Gunsilius, %
    Peter Hull, %
    Arthur Lewbel, %
    Daniel Lewis, %
    Dexin Li, %
    Adam McCloskey, %
    Geert Mesters, %
    Aureo de Paula, %
    Julius Sch\"{a}per, %
    Christian Wolf, %
    participants at the JBES session at the 2025 ASSA meetings
    and the Federal Reserve Bank of New York seminar, %
    two anonymous referees, the editor Atsushi Inoue, and particularly from our discussants: S\'{\i}lvia Gonçalves, %
    Ed Herbst, %
    Ana Mar\'\i a Herrera, %
    Ben Johannsen, %
    Òscar Jordà, %
    Toru Kitagawa, %
    Elena Pesavento, %
    Weining Wang, %
    and Mengshan Xu. Eric Qian and Nelson Matthew Tan provided excellent research assistance.
    Koles\'{a}r acknowledges support by the National Science Foundation under
    Grant SES-2049356. Plagborg-M{\o}ller acknowledges support from the National
    Science Foundation under Grant SES-2238049
    and from the Alfred P.\ Sloan Foundation.}} \author{Michal Koles\'{a}r \\
  Princeton University \and Mikkel Plagborg-M{\o}ller \\ Princeton University}
\date{\texorpdfstring{\bigskip}{ }\today}
\maketitle

\vspace{-1.5em}
\begin{abstract}
Applied macroeconomists frequently use impulse response estimators motivated by linear models. We study whether the estimands of such procedures have a causal interpretation when the data generating process is in fact nonlinear. We show that vector autoregressions and linear local projections onto observed shocks or proxies identify weighted averages of causal effects regardless of the extent of nonlinearities. By contrast, identification approaches that exploit heteroskedasticity or non-Gaussianity of latent shocks are highly sensitive to departures from linearity. Our analysis is based on new results on the identification of marginal treatment effects through weighted regressions, which may also be of interest to researchers outside macroeconomics.
\end{abstract}

\section{Introduction}\label{sec:intro}

Impulse response functions are key objects in macroeconomic analysis. Since they
measure dynamic causal effects of surprise changes in policy or fundamentals on
subsequent macroeconomic outcomes, they provide calibration targets for
structural modeling and help validate model predictions. They also inform
optimal economic policy questions, both directly and indirectly
\citep{Christiano1999,McKay2023}.

Applied researchers typically report impulse response estimators motivated by
linear time series models, such as \acp{VAR} or local projections. Although
there exists a wealth of nonlinear alternatives (\citealp{Fan2003};
\citealp{Herbst2016}; \citealp[Chapter 18]{Kilian2017}), linear methods are
attractive due to their simplicity and the difficulty of clearly detecting
nonlinear relationships in typical macroeconomic data. At the same time, both
macroeconomic theorists and policymakers think nonlinearities are important:
structural models with essential nonlinearities have become dominant in recent
decades, and many economic policy debates concern state-dependence and
asymmetries. How can we justify using linear methods if we think the world is a
nonlinear place?

This paper studies the causal interpretation of impulse response estimators
based on linear models when the data is generated by an essentially unrestricted
nonparametric structural model. We first deliver \emph{good} news for linear
local projection or \ac{VAR} estimators that project directly on an observed
shock or proxy: their estimand (i.e., probability limit) equals a weighted
average of the true nonlinear causal effects, regardless of the extent of
nonlinearities in the \ac{DGP}. By contrast, the news is \emph{bad} or even
\emph{ugly} for estimators that identify latent shocks via heteroskedasticity or
non-Gaussianity: they generally do not estimate a meaningful causal summary
under departures from linearity. Thus, the hard work needed to directly measure
shocks (or proxies) using historical or institutional data buys insurance
against nonlinearities that other identification approaches lack.

Our good news are based on an extension of the results in \citet{Yitzhaki1996}
and \citet{Rambachan2021}: impulse response estimands from linear local
projections and \acp{VAR} that project on observed shocks or proxies correspond
to positively-weighted averages of marginal effects---causal effects of changing the shock variable from $x$ to $x+\delta$ for infinitesimally small $\delta$, averaged out over all (past, present, and future) shocks other
than the contemporaneous shock of interest, and weighted over different baseline
values $x$ for the shock variable. Thus, these estimands provide a scalar causal summary of
the full richness of the nonlinear causal effects; the positive weights ensure
that the researcher gets the sign right if the true marginal effects are
uniformly positive or negative. Our assumptions drop restrictions imposed in the
existing literature that ruled out models with kinks or discontinuous
regime-switches or shocks with unbounded support.

In a nonlinear \ac{DGP}, both the sign and the magnitude of the causal effects
can depend on the baseline shock value (e.g., whether it is positive or
negative), so how these values are weighted can matter a lot. Fortunately, as we
illustrate using several empirical examples, the weight function used by local
projections and \acp{VAR} is straightforward to estimate and report. In many
applications, the researcher does not directly observe the shock of interest but
only a proxy, also known as an external instrument \citep{Stock2018}. In this
case, we show that an easily-interpretable monotonicity condition is required to
guarantee a positive weight function. However, we also show that when control variables are needed to isolate a true shock (i.e., recursive or
Cholesky identification), then positive weights can only be guaranteed if the linearly residualized shock is \emph{nonlinearly} unpredictable by the controls, which may be a strong assumption in the absence of detailed institutional knowledge and high-quality data.

One implication of these results is that linearity-based estimators are useful
even when economic theory predicts a nonlinear relationship between the shock
and the outcome of interest. For example, if the outcome variable has limited
support, such as when it is binary or censored (say, due to a zero lower bound),
nonlinearities are inherently present. If one is interested in
\emph{characterizing} the nonlinearities, then it makes sense to model them, and
it is of course always a good idea to plot the raw data regardless. However, if
one is interested in an overall summary of \emph{marginal effects}, then linear
local projections and \acp{VAR} are theoretically coherent estimators, as
discussed earlier. In fact, we show that directly modeling nonlinearities can be
counterproductive unless the researcher is confident in their modeling: under
functional form misspecification, local projections with higher-order terms
still estimate a weighted average of marginal effects, but some of the weights
may be negative, which risks getting the sign of the causal effects wrong. This
echoes the message from an earlier JBES lecture by \citet{angrist_dummy_2001}
that in a cross-section context with limited dependent variables, linear methods
provide more robust estimates of treatment effects than nonlinear ones.

When there is a dearth of direct shock measures or proxies, applied researchers
frequently resort to identification via heteroskedasticity
\citep{Sentana2001,Rigobon2003,Lewbel2012}. Unfortunately, we show that these
estimation approaches are sensitive to the assumption that the structural model
is linear: the estimand can easily be nonzero when there is no causal effect, or
negative when the true shock has a uniformly positive effect on the outcome of
interest. Fixing these issues while still delivering informative inference
appears difficult, since a natural nonparametric generalization of the
identification strategy yields very wide identified sets. The intuition for
these negative results is that the identification exploits a source of exogenous
variation that shifts the \emph{scale} of the latent shock of interest but not
its mean. Without strong functional form assumptions, this type of exogenous
variation is uninformative about the effect of a \emph{location} shift in the
shock on the conditional mean of the outcome, i.e., the impulse response.
However, a silver lining is that the linear model delivers testable restrictions.

The sensitivity to nonlinearity is even greater for identification via non-Gaussianity \citep{Comon1994,Gourieroux2017,Lanne2017}. Also known as \ac{ICA}, this identification approach has recently increased significantly in popularity in the \ac{VAR} literature. We show that the nonparametric analogue of the identification assumptions yields an identified set so large that effectively any function of the data can be construed as a ``shock''. Intuitively, the mere assumptions that the latent shocks are independent and non-Gaussian are vacuous in a nonparametric context: any collection of random variables can always be represented as some nonlinear function of independent uniformly distributed random variables. Moreover, we give examples of simple \acp{DGP} featuring slight nonlinearity for which \emph{any} linearity-based \ac{ICA} procedure is highly biased asymptotically, yet in these \acp{DGP} one cannot reject the validity of the linear model.

The building block underlying most of the above findings is a set of results on
the identification of weighted averages of marginal treatment effects using
weighted regressions, which connects our analysis to a large literature in
microeconometrics \citep[e.g.,][]{Yitzhaki1996,NeSt93,Angrist1999,GoldsmithPinkham2024}.
We extend existing results in this literature by unifying the treatment of
continuous, discrete, and mixed regressors, and by substantially weakening the
regularity conditions: we allow for regressors with unbounded support, impose
minimal regularity on the regression function, and our moment conditions
essentially only require the existence of the probability limit of the
regression estimator.

An important limitation of our results is that they only concern identification.
While we are motivated by the observation that full-fledged nonparametric
estimation is challenging in realistic macroeconomic data sets, we do not
explicitly analyze the precision or small-sample bias of the estimators we
study. We refer to \citet{Herbst2024} for a discussion of finite-sample biases
of local projections and \acp{VAR} in linear models. Another limitation is that we do not consider identification via non-recursive short-run restrictions, long-run restrictions, or sign restrictions.

\paragraph{Literature.}
Pioneering work on semiparametric causal time series analysis includes
\citet{Gallant1993}, \citet{Potter2000}, \citet{White2006}, \citet{White2009}, \citet{Angrist2011},
and \citet{Angrist2018}, see also \citet{Goncalves2021,goncalves24nonparametric,Goncalves2024},
\citet{Gourieroux2023}, and \citet{Kitagawa2023} for recent contributions. Our
result on the causal interpretation of local projections with observed shocks is
very closely related to \citet{Rambachan2021} and subsequent work by
\citet{Caravello2024} and \citet{Casini2024}.

As for identification via heteroskedasticity or non-Gaussianity, we are not
aware of other work in a nonparametric vein. \Citet{MontielOlea2022} criticize
linearity-based versions of these identification strategies for being seemingly
sensitive to functional form assumptions, and potentially being subject to weak
identification. The present analysis quantifies this sensitivity more precisely
by deriving both the identified sets for the nonparametric analogues of these
identification assumptions, and the estimands of linearity-based procedures.

\paragraph{Outline.}
\Cref{sec:framework} defines a nonparametric framework for identification of
dynamic causal effects. \Cref{sec:shock} argues that local projection and \ac{VAR} estimands
based on observed shocks or proxies have a robust causal interpretation
regardless of the extent of nonlinearities. \Cref{sec:heterosk,sec:nongauss}
show, on the other hand, that estimands based on identification through
heteroskedasticity or non-Gaussianity are sensitive to the assumption that the
structural function is linear. \Cref{sec:mte} provides the theoretical basis for
the results in the earlier parts of the paper by extending results from the
microeconometric literature on the interpretation of regression estimators as
weighted marginal treatment effects; this section may be of independent interest
for readers outside macroeconomics. \Cref{sec:concl} concludes. Technical
details and proofs are relegated to the online supplement.

\section{Nonparametric framework for dynamic causality}\label{sec:framework}

In this section we set up a nonparametric framework for dynamic causal
identification.

\subsection{Model}

We are interested in the dynamic response of a scalar outcome variable $Y_{t}$
to an impulse in the scalar shock variable $X_t$. As a leading example, one may
think of $X_t$ as a variable controlled by a policy-maker, such as a surprise
change in the policy interest rate set by the central bank. For ease of
exposition, we restrict attention to continuously distributed shocks $X_t$ for
now, but \Cref{sec:mte} shows that our results generalize to handle continuous,
discrete, or mixed distributions in a unified manner.

The outcome variable is determined by an underlying dynamic structural model.
Our causal framework doesn't restrict this model; we only assume that when
evaluated $h$ periods after the realization of the shock $X_{t}$, the outcome
admits the nonparametric structural representation
\begin{equation} \label{eqn:causal}
Y_{t+h} = \psi_h(X_t, \vec{U}_{h, t+h}) \quad \text{for all } t, h \geq 0.
\end{equation}
For each horizon $h$, $\psi_{h}(\cdot, \cdot)$ is an unknown measurable function
that we call the \emph{structural function}, while $\vec{U}_{h, t+h}$ is a
vector of all variables (dated before, on, and after time $t$) that causally
affect $Y_{t+h}$, other than $X_t$. Without restrictions on
$\vec{U}_{h, t+h}$ or the structural function, the
representation~\eqref{eqn:causal} is without loss of generality. In typical
recursive time-series models, however, $\vec{U}_{h, t+h}$ will contain the vector $\vec{Y}_{t-1}$ of observed data at time
$t-1$ as well as shocks dated $t, t+1,\dotsc, t+h$, but exclude $X_t$ and shocks
dated after $t+h$ \citep{White2006,White2009,Caravello2024,Goncalves2024}. A
leading special case is the linear structural \ac{VAR} model, which additionally
implies that $\psi_{h}$ is linear in both $X_{t}$ and $\vec{U}_{h, t+h}$
\citep[e.g.,][Chapter 4.1]{Kilian2017}.

We assume throughout that it is meaningful to think of varying $X_t$ while
keeping $\vec{U}_{h, t+h}$ constant, so that the random function
$Y_{t+h}(x)\equiv \psi_h(x, \vec{U}_{h, t+h})$ defines a \emph{potential outcome
  function} at horizon $h$. \Citet{Angrist2011}, \citet{Angrist2018}, and
\citet{Rambachan2021} work with this potential outcome notation
$Y_{t+h}(x)$, and keep all other past, present, and future shocks
(captured by $\vec{U}_{h, t+h}$ in our model) implicit. Our structural function
framework~\eqref{eqn:causal} is mathematically equivalent, but facilitates comparisons with the
linear structural \ac{VAR} literature.

For now, we focus on the case where $X_{t}$ is a
``shock'' (such as a surprise change in a policy
instrument), and assume that it is independent of $\vec{U}_{h, t+h}$:
\begin{equation}\label{eqn:indep}
X_t \independent \vec{U}_{h, t+h}.
\end{equation}
This assumption is common in the literature, and essentially just normalizes the
structural function $\psi_{h}$, so that its first argument captures the total
causal effect of the shock $X_t$ on $Y_{t+h}$, including its direct effect and
any indirect effects, both contemporaneous and dynamic. This is illustrated in
the following simple example.

\begin{exm}\label{exm:ar1_regime}
Consider a univariate AR(1) model with endogenous regime switching:
\begin{equation*}
  Y_t = \rho_t Y_{t-1} + \tau \varepsilon_t + \nu_t,
\end{equation*}
with regime-dependent parameter $\rho_t = \rho_1 S_t + \rho_0(1-S_t)$ and binary
regime $S_t = \1{\varepsilon_{t-1}+\xi_{t-1}\leq 0}$, and where
$\rho_0,\rho_1,\tau$ are constants. Assume that $\varepsilon_t$, $\nu_t$, and
$\xi_t$ are i.i.d.\ and mutually independent, and that we observe the shock
$X_t=\varepsilon_t$. We can cast this model into the form required by
\cref{eqn:causal,eqn:indep} as follows. Define
$\vec{U}_{h, t+h} \equiv (Y_{t-1},S_t, \nu_t, \dotsc, \nu_{t+h}, \xi_t, \dotsc,
\allowbreak\xi_{t+h-1}, \varepsilon_{t+1}, \dotsc, \varepsilon_{t+h})'$ and
$\rho(\vartheta) \equiv \rho_1 \1{\vartheta \leq 0} + \rho_0\1{\vartheta>0}$ for
$\vartheta \in \mathbb{R}$. Then, for all $h \geq 1$,
\begin{align*}
  \psi_h(x, \vec{u}) &= \big\lbrace y_{-1}(\rho_1 s+ \rho_0(1-s)) + (\tau x+\nu)\big\rbrace \rho(x+\xi)\prod_{\ell=1}^{h-1} \rho(\varepsilon_{+\ell}+\xi_{+\ell}) \\
                     &\quad + \sum_{\ell=1}^h (\tau\varepsilon_{+\ell}+\nu_{+\ell})\prod_{b=\ell}^{h-1} \rho(\varepsilon_{+b}+\xi_{+b}),
\end{align*}
where we have partitioned
$\vec{u}=(y_{-1}, s, \nu, \nu_{+1}, \dotsc, \nu_{+h}, \xi, \xi_{+1}, \dotsc,
\xi_{+(h-1)}, \varepsilon_{+1}, \dotsc, \varepsilon_{+h})'$. Notice that the
function $\psi_h$ captures the full dynamic effect of the shock variable
$X_t=\varepsilon_t$: both the direct impact effect of $\varepsilon_t$ on $Y_t$
(which feeds forward to future periods), and the indirect, nonlinear effect of
$\varepsilon_t$ on the next-period regime $S_{t+1}$ (which also feeds forward).
\end{exm}

When $X_{t}$ is not a shock but corresponds to a policy instrument (such as the policy interest
rate) or any other kind of serially correlated variable, it will typically be correlated with $\vec{U}_{h, t+h}$, both
because of the serial correlation and because it may be correlated with
other determinants of $Y_{t+h}$---we discuss this case
in~\Cref{sec:ident-with-contr}.

\subsection{Causal effects}
A familiar issue in nonlinear models is that there are multiple possible
definitions of an impulse response, i.e., a dynamic causal effect. In a linear
model, the effect of exogenously changing $X_{t}$ from $x_{0}$ to $x_{1}$ is a
linear function of the difference $x_{1}-x_{0}$: it equals
$\psi_h(x_1,\vec{U}_{h, t+h})-\psi_h(x_0,\vec{U}_{h,
  t+h})=\theta_h(x_{1}-x_{0})$ for some constant scalar $\theta_{h}$. By
contrast, in a nonlinear model, the effect is in general a nonlinear function
of both $x_{1}$ and $x_{0}$ (not just their difference), and it also depends on
the past history and the current and future nuisance shocks via
$\vec{U}_{h, t+h}$.

In theoretical macroeconomic modeling, researchers often report the impulse
responses with respect to a so-called ``MIT shock'', which starts the economy at
steady state, then hits the economy with a one-off impulse to $X_t$, and
subsequently sets all other current and future shocks to zero:
$\psi_h(X_{t}, \vec{0})-\psi_h(0,\vec{0})$, where we normalize the steady-state
values of $X_{t}$ and $\vec{U}_{h, t+h}$ to zero. While computationally
convenient, this impulse response concept has no empirical counterpart; moreover, for the
measure to be policy relevant, it is necessary that the model satisfy certainty
equivalence.

We shall instead focus on impulse responses (causal effects) defined as expected counterfactual changes in the outcome of interest, averaging out over all other shocks. Specifically, define the \emph{average structural function}
\begin{equation*}
  \Psi_h(x) \equiv E[\psi_h(x, \vec{U}_{h, t+h})], \quad x\in\mathbb{R},
\end{equation*}
which corresponds to the expected potential outcome function. Here the
expectation is taken over the marginal distribution of $\vec{U}_{h, t+h}$. The
expectation is implicitly assumed to exist for all $x$. The average structural
function measures the counterfactual average value of the future outcome
$Y_{t+h}$ that we would observe if the policy-maker engineered a particular
fixed value $x$ for the policy variable at time $t$, averaging out over the
randomness caused by all other factors that influence the outcome independently
of the policy decision at time $t$. Even though in some nonlinear models the
structural function $\psi_h$ is discontinuous in the policy variable $x$, the
average structural function $\Psi_h$ will typically be a smoother function of
$x$, as it averages out over the realizations of other shocks. For example, this
is the case in the regime-switching model in \Cref{exm:ar1_regime} if $\xi_t$ is
continuously distributed. A certain amount of smoothness in $\Psi_h$ will be
important for the identification of causal effects, as we discuss below.

Typical data samples in macroeconomics are too small to permit accurate
nonparametric estimation of the entire average structural function
$x \mapsto \Psi_h(x)$. A pragmatic alternative is to target weighted averages of
the structural function---average causal effects---or its derivatives---average
marginal effects \citep{Rambachan2021,Goncalves2021,Goncalves2024}. This paper
focuses on estimation of \emph{average marginal effects}
\begin{equation}\label{eqn:theta}
\theta_h(\omega)\equiv\int \omega(x)\Psi_h'(x)\, dx,
\end{equation}
where $\omega(\cdot)$ is a weight function averaging across the baseline values
of the shock variable $X_{t}$. We reserve the term \emph{average marginal
  effect} to weight functions that are convex, i.e., $\omega(x)$ is nonnegative
for all $x$ and integrates to one, $\int \omega(x)\, dx=1$. This ensures that
$\theta_{h}(\omega)$ is a meaningful causal summary of the average structural
function $\Psi_{h}(x)$ in that it prevents what \citet{strlb17} call a
sign-reversal: if $\Psi_{h}'(x)$ has the same sign for all $x$ ($+,0$ or $-$),
then $\theta_{h}(x)$ will also have this sign.\footnote{\citet{blandhol2022tsls}
  call estimands with convex $\omega$ ``weakly causal''. Convex weighting
  schemes also satisfy what \citet{rslr07} call ``boundedness'':
  $\theta_{h}(\omega)$ lies in the support of $\Psi_{h}'(x)$.} This property is
particularly useful when qualitatively validating predictions of structural
macroeconomic models.

Depending on the form of the weight function,
$\theta_{h}(\omega)$ has two interpretations in terms of an \emph{average causal
  effect} of a shock with magnitude $\delta>0$,
\begin{equation} \label{eqn:avg_causal_effect}
\theta_h(\delta, \omega_{0}) \equiv \frac{1}{\delta}\int \omega_{0}(x) \lbrace \Psi_h(x+\delta)-\Psi_h(x) \rbrace \, dx.
\end{equation}
First, the average marginal effect corresponds to the average causal effect for
infinitesimally small shocks: $\theta_{h}(\omega)=\lim_{\delta\to 0}\theta_{h}(\delta, \omega)$, provided we can pass the limit as $\delta\to 0$
under the integral sign in~\eqref{eqn:avg_causal_effect}. Second, if the
weighting in~\eqref{eqn:theta} admits the integral representation
$\omega(x)=\frac{1}{\delta}\int_{x-\delta}^{x}\omega_{0}(x)\, dx$, substituting
$\Psi_h(x+\delta)-\Psi_h(x)=\int_{x}^{x+\delta}\Psi'(\chi)\, d\chi$
into~\eqref{eqn:avg_causal_effect} and changing the order of integration yields
$\theta_{h}(\omega)=\theta_{h}(\delta, \omega_{0})$. For this reason, focusing on
average marginal effects is without loss of generality.

In a linear model, the weighting does not matter, since $\Psi_{h}'(x)$ does not
depend on $x$. But in nonlinear models, it could matter greatly whether we
attach most weight to positive or negative shocks, or to shocks with small or
large magnitude. Therefore, accounting for the form of the weighting $\omega$ is
important when using estimates of $\theta_{h}(\omega)$ to calibrate or validate
structural macroeconomic models. In the next section, we discuss identification
approaches that deliver weighted averages of marginal effects under a particular
weighting scheme that depends on the shock distribution. In \Cref{sec:mte}, we
discuss estimation approaches that target any pre-specified weighting scheme.

\section{The good: observed shocks and proxies}\label{sec:shock}

If the researcher directly observes the shock of interest, or at least a valid
proxy for it, then there is \emph{good} news: conventional local projections or
structural \ac{VAR} impulse responses estimate average marginal effects with an
interpretable weighting scheme, regardless of how nonlinear the underlying
\ac{DGP} is. Moreover, the weights can be estimated from the data, and we give
several empirical examples of how to interpret them. In contrast to
linear estimators, we demonstrate using a simple example that nonlinear
extensions of local projections or \acp{VAR} do not generally provide
meaningful causal summaries under misspecification. Finally, we extend the
analysis to shocks that are recursively identified, i.e., by controlling for
covariates.

\subsection{Identification with observed shocks}\label{sec:obsshock}

We start off by assuming that the researcher directly observes (or consistently estimates) the shock $X_t$ of interest. This would be the case, for example, if the shock is identified through a ``narrative approach''. See \citet{Ramey2016} for several empirical examples.

Under the nonlinear structural model~\eqref{eqn:causal} and the shock
independence assumption~\eqref{eqn:indep}, the conditional expectation of the
outcome given the shock,
\begin{equation} \label{eqn:cond_expn}
g_h(x) \equiv E[Y_{t+h} \mid X_t=x],
\end{equation}
nonparametrically identifies the average structural function:
\begin{equation} \label{eqn:obsshock_identif}
\Psi_h(x) = E[\psi_h(x, \vec{U}_{h, t+h})] = E[\psi_h(x, \vec{U}_{h, t+h}) \mid X_t=x] = g_h(x).
\end{equation}
Hence, in principle, we could estimate any weighted causal effect of interest by
running a nonparametric regression of $Y_{t+h}$ on $X_t$ to obtain $g_h(\cdot)$
in the first step, and then averaging this function according to the desired
weighting scheme in the second step, as suggested by \citet{goncalves24nonparametric,Goncalves2024}. In \Cref{sec:mte}, we discuss a complementary strategy that
identifies the same estimand via weighted averages of the observed outcomes, and
how both strategies can be combined. However, as discussed in more detail in
\Cref{sec:mte}, these strategies may yield noisy and sensitive estimates in the
relatively small samples available in macroeconomics.

\paragraph{Interpretation of linear projection estimates.}
We take a cue from \citet{Rambachan2021} and instead aim for a less ambitious
goal. Rather than targeting a pre-specified weighted average of
causal or marginal effects, we focus on simple local projection
and \ac{VAR} estimators, which are relatively precise even with small sample sizes. We
demonstrate that these simple estimators have an attractive robustness property:
even though they are motivated by a linear model, when the \ac{DGP} is nonlinear,
their estimand can still be interpreted as an average marginal effect with a
particular weight function.

The local projection estimator of \citet{Jorda2005} estimates the impulse
response of $Y_t$ with respect to $X_t$ at horizon $h$ as the coefficient
$\hat{\beta}_h$ in the \ac{OLS} regression
\begin{equation} \label{eqn:lp_reg}
Y_{t+h} = \hat{\beta}_h X_t + \hat{\gamma}_h'\vec{W}_t + \text{residual}_{h, t+h},
\end{equation}
where $\vec{W}_t$ is a vector of control variables (typically including a
constant and lagged outcomes and shocks). For now, we will assume that the shock
$X_t$ is in fact a ``shock'', so that it is linearly unpredictable using the
controls: $\cov(X_t, \vec{W}_t)=0$. Then the set of controls
$\vec{W}_t$ affects only the precision of $\hat{\beta}_{h}$, but not its
probability limit. In particular, under standard stationarity and ergodicity
assumptions, the local projection estimator $\hat{\beta}_h$ will converge in
probability to the population projection coefficient
\begin{equation} \label{eqn:lp_estimand}
\beta_h \equiv \frac{\cov(g_{h}(X_{t}), X_t)}{\var(X_t)}.
\end{equation}
\citet[Propositions 1 and 2]{PMW2021} show that a \ac{VAR} which includes $X_t$
ordered first has the exact same population estimand~\eqref{eqn:lp_estimand},
provided that the number of lags in the \ac{VAR} is sufficiently
large.\footnote{If $X_t$ is linearly unpredictable from lagged data, it is
  sufficient that the lag length weakly exceed $h$.} It is a textbook result
that the linear function $\beta_{h}x$ provides the best \emph{linear}
approximation to the potentially nonlinear average structural function
$g_{h}(x)=\Psi_{h}(x)$ \citep[e.g.,][Theorem 3.1.6]{Angrist2009}, so that it
approximates the average structural function in a \emph{prediction} sense. However,
this result is not directly informative about whether $\beta_{h}$ has a
\emph{causal} interpretation---whether it can be interpreted as an average
marginal effect if $\Psi_{h}(x)$ is nonlinear.

The following proposition shows that the local projection and \ac{VAR}
estimand~\eqref{eqn:lp_estimand} achieves our goal: it has a causal
interpretation as an average marginal effect~\eqref{eqn:theta}. The result is
not new---it appeared previously in \citet{Yitzhaki1996} and
\citet{Rambachan2021}; as we discuss below, the novelty lies in substantively
weakening the regularity conditions.

\begin{prop}\label{thm:obsshock} Assume that $X_t$ is continuously distributed
  on an interval $I\subseteq\mathbb{R}$ (the interval may be unbounded, and
  could equal $\mathbb{R}$), with positive and finite variance. Assume that the
  conditional mean $g_h$ defined in~\eqref{eqn:cond_expn} is locally absolutely
  continuous on $I$.\footnote{That is, absolutely continuous on any compact
    interval contained in $I$.} Suppose finally that
  $E[\abs{g_h(X_{t})}(1+\abs{X_{t}})]<\infty$ and
  $\int_I \omega_X(x)|g_h'(x)|\, dx<\infty$, where
  \begin{equation}\label{eqn:obsshock_weights}
    \omega_X(x) \equiv \frac{\cov(\1{X_t \geq x}, X_t)}{\var(X_t)}.
  \end{equation}
  Then the estimand~\eqref{eqn:lp_estimand} satisfies
  \begin{equation*}
    \beta_{h} = \int_{I} \omega_X(x)g_h'(x)\, dx,
  \end{equation*}
  and the weight function $\omega_X$ has the following properties:
\begin{enumerate}[(i)]
\item It is convex: $\omega_{X}(x)$ is non-negative for all $x$, and integrates
  to one, $\int_I \omega_X(x)\, dx=1$.
\item It is hump-shaped: monotonically increasing from 0 to its maximum for
  $x \leq E[X_t]$, and then monotonically decreasing back to 0 for
  $x \geq E[X_t]$.
\item It depends only on the marginal distribution of $X_t$, and not on the
  conditional distribution of $Y_{t+h}$ given $X_{t}$.
\end{enumerate}
\end{prop}
Combined with the identification result~\eqref{eqn:obsshock_identif} for the
average marginal effect, \Cref{thm:obsshock} shows that linear local projections
and \acp{VAR} remain useful in a nonlinear world: they estimate an average
causal effect $\theta_h(\omega_X) = \int \omega_X(x)\Psi_h'(x)\, dx$ for
infinitesimal shocks, with a convex weighting scheme $\omega_{X}$. Furthermore,
the scheme gives most weight to shocks close to the mean $E[X_t]$, with little
weight given to extreme values. In the special case where $X_{t}$ is normally
distributed, \Cref{thm:obsshock} reduces to Stein's lemma \citep[Lemma 1
in][]{stein81}: the weight function $\omega_X$ reduces to the normal density function, so
that $\beta_{h}$ equals the expected marginal effect, $E[\Psi_{h}'(X_{t})]$, as
noted by \citet{Yitzhaki1996}. The fact that the weighting scheme
depends only on the marginal distribution of $X_{t}$ and not the particular
outcome variable $Y_{t+h}$ or horizon $h$ allows for comparisons of average marginal
effects for different outcomes or across different horizons $h$.\footnote{Since the weighting does not depend on the outcome variable or horizon, multipliers (i.e., ratios of cumulative impulse response estimands, see \citealp[Section 3]{Jorda2023}) can be expressed as a ratio of two average marginal effects with the same weight function, where we now additionally average across horizons.} If the true
\ac{DGP} is in fact linear, then the weighting of course does not matter, and we
recover the conventional linear impulse response.

While we focus here on interpreting the proposition in the context of the causal
model in \Cref{sec:framework}, the result does not require the structural
assumptions~\eqref{eqn:causal}--\eqref{eqn:indep}. This is relevant in settings
in which the conditional mean $g_h(x)=E[Y_{t+h} \mid X_t=x]$ is a useful
descriptive object even if it does not have a direct causal interpretation.

The assumption that $X_{t}$ is continuously distributed can be dropped without
changing the result, as we show in \Cref{sec:mte}. In cases where there are gaps
in the support of $X_{t}$, such as when the shock is discrete or mixed, one just
needs to extend the definition of the conditional mean function $g_{h}$ to the
whole interval $I$ by linear interpolation. To our knowledge, this unification
of the treatment of continuous, discrete, and mixed distributions is novel.

Even in the case of a continuously distributed shock, the assumptions in \Cref{thm:obsshock} are substantively weaker than those in the
literature, and accommodate all textbook linear models as well as a wide range
of nonlinear models. The assumption that $g_{h}$ is locally absolutely
continuous is necessary to ensure that weighted marginal effects are
well-defined. As discussed earlier, this assumption will typically hold even in
models with discrete regimes or kinks, since the
expectation~\eqref{eqn:obsshock_identif} averages over the distribution of the nuisance
shocks. The moment conditions and integrability condition
$\int_{I} \omega_X(x)\abs{g_h'(x)}\, dx<\infty$ just ensure that the estimand
$\beta_{h}$ and the weighted marginal effect
exist.\footnote{\Cref*{lemma:finite_integral_omega} in \Cref*{app:proofs} shows
  that for the integrability condition to hold, it is sufficient to assume the
  tails of $g_{h}(x)$ are monotone.} In contrast, the original work by
\citet{Yitzhaki1996} does not provide a formal proof or regularity conditions on
$g_{h}$ (neither does the discussion by \citealp[pp.\ 78 and 110]{Angrist2009}).
Analogous results in \citet{Rambachan2021}, \citet{Graham2022},
\citet{Caravello2024}, and \citet{Casini2024} require the potential outcome
function (not its expectation) to be smooth, which rules out models with kinks
or discrete regimes, and require the interval $I$ to be bounded, which rules out
the textbook case of normally distributed shocks.\footnote{Our result is not strictly more general than that of \citet{Casini2024}, since they allow for non-stationary potential outcomes.} The restrictiveness of these
conditions led \citet[Appendix C]{Goncalves2024} to question the applied
relevance of the causal interpretation of the estimand~\eqref{eqn:lp_estimand},
but our weaker conditions demonstrate that this concern is unfounded when interest centers on average marginal effects.

\paragraph{Estimating the weight function.}
As argued by \citet{Angrist1999} for the case of discrete $X_t$, the weight function $\omega_X$ defined in~\eqref{eqn:obsshock_weights} can be estimated in the data. This allows the researcher to gauge which weighted causal effect is being estimated: does it attach most weight to negative or positive shocks, small or large shocks? Since the weight function depends only on the shock variable itself and not the outcome variable or the impulse response horizon, it is only necessary to estimate a single function. We therefore recommend that researchers always estimate and plot this function.

Estimation is simple: $\omega_X(x)$ equals the
slope coefficient in a (population) regression of the indicator $\1{X_t \geq x}$
on $X_t$. This regression can be implemented in the data via \ac{OLS}, separately for each value $x$ of
$X_t$ observed in the data.\footnote{Pointwise confidence intervals can be obtained with conventional heteroskedasticity-robust standard errors. One could also use autocorrelation robust
  standard errors to allow for time series dependence of $X_t$, but causal
  interpretation is more challenging if the shocks are not independent.} In
applications, it may also be of interest to report an integral
$\int_{\underline{x}}^{\overline{x}} \omega_X(x)\, dx$ of the weight function
over an interval $x \in [\underline{x}, \overline{x}]$. \Cref*{app:weight_integr}
shows that we can estimate this integral by the slope coefficient in an OLS
regression of
$M_t \equiv \max\lbrace \min\lbrace X_t, \overline{x} \rbrace, \underline{x}
\rbrace$ on $X_t$.
In particular, to estimate the total weight $\int_0^\infty \omega_X(x)\, dx$
given to positive shocks, we simply regress
$M_t \equiv \max\lbrace X_t,0\rbrace$ on $X_t$.

To illustrate, we now empirically estimate the weight function $\omega_X$ for
various macroeconomic shocks considered in the handbook chapter by
\citet{Ramey2016}. We use \citeauthor{Ramey2016}'s replication code and data off
the shelf. In particular, prior to computing weights, all shocks are
residualized on the same control variables that she uses in her VARs and local
projections. The estimates of the weight functions are
obtained from OLS regression output, as described above. To demonstrate the ease of implementation, all steps of the
computations are carried out in Stata, like \citeauthor{Ramey2016}'s replication
code.\footnote{Our code and data are available at \url{https://github.com/mikkelpm/nonlinear_dynamic_causal}}

\Cref{fig:gov} shows the estimated weight functions for four identified
government spending shocks from the applied literature. Note that the shocks are
not entirely comparable due to differences in their precise definitions and
sample periods. The \citet{Blanchard2002} and \citet{Fisher2010} shocks, which
are intended to capture general government spending shocks, yield approximately
symmetric weight functions. By contrast, the \citet{BenZeev2017} and
\citet{Ramey2011} shocks, which capture news about future defense spending,
generate weight functions that are skewed towards positive shocks.
In fact, both these shocks exhibit a large positive outlier in the 3rd quarter of
1950 (the onset of the Korean War), reflected in the fat right tail of the
weight functions. In other words, impulse responses from local
projections or \acp{VAR} estimated off the latter two shocks will largely
reflect the causal effects of sharp military \emph{buildups}, rather than
retrenchments. This is important to remember when using empirical impulse
responses to discipline structural models that feature asymmetries (such as
downward nominal wage rigidity or borrowing constraints), since then
model-implied impulse responses with respect to positive government spending
shocks will differ from those for negative shocks. \Cref*{app:weight_empir} gives
further examples of weight functions for several identified tax, technology, and
monetary policy shocks. As these examples illustrate, plotting estimates of the
weights $\omega_{X}$ is useful in interpreting the results of any subsequent
impulse response analysis and for comparing with prior studies.

\begin{figure}[!t]
\centering
\textsc{Empirical weight functions: government spending shocks} \\
\includegraphics[width=\linewidth,clip=true,trim=0 1em 0 0]{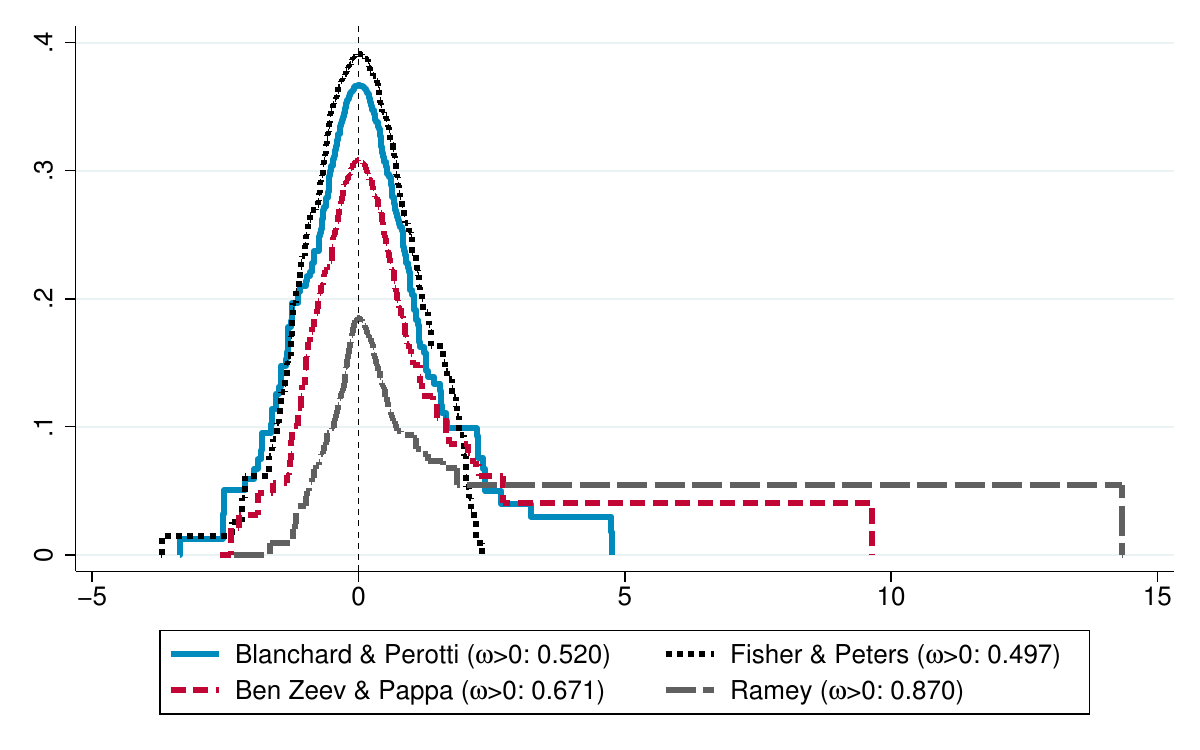}
\caption{Estimated causal weight functions $\omega_X$ for government spending shocks obtained from the replication files for \citet{Ramey2016}, quarterly data. Horizontal axis in units of standard deviations. ``$\omega>0$'': total weight $\int_0^\infty \omega_X(x)\, dx$ on positive shocks. Papers referenced: \citet{Blanchard2002}, \citet{Fisher2010}, \citet{BenZeev2017}, \citet{Ramey2011}.}\label{fig:gov}
\end{figure}

\paragraph{Parametric nonlinear specifications.}
In many cases, economic theory predicts that the average structural function
$\Psi_{h}$ is likely nonlinear. For example, if the outcome variable has limited support,
such as due to censoring or when it is discrete, the structural function must
necessarily be nonlinear. In such cases, it seems natural to model the
nonlinearity directly, rather than to stick to a linear specification as
in~\eqref{eqn:lp_reg}. For example, \citet{Jorda2005} and \citet{Jorda2025}
suggest including powers of the shock in local projections. Similarly, there is
a rich literature on nonlinear extensions of \ac{VAR} models, see for example
\citet[Chapter 18]{Kilian2017}. Such direct modeling of the nonlinearities is
sensible if the goal of the analysis is to directly characterize the extent and
types of nonlinearity present in the data, e.g., threshold effects or sign and
size dependence \citep{Caravello2024}.

However, for estimating average causal effects, simple linear local projections
or \acp{VAR} appear more robust than nonlinear parametric specifications. As
shown in \Cref{thm:obsshock}, the linear specification in~\eqref{eqn:lp_reg} is
robust to misspecification in that it estimates a well-defined average marginal
effect regardless of the form of nonlinearity in the structural function
$\Psi_{h}$. By contrast, we now show that this is not the case for a local
projection specification that includes a quadratic term, echoing similar results
in \citet{angrist_dummy_2001} regarding robustness of parametric nonlinear
limited dependent variable models. These results suggest that nonlinear
specifications do not generally have such a robustness property.

Consider a quadratic local projection of $Y_{t+h}$ on $X_t$, $X_t^2$, and an
intercept.\footnote{While we focus on the quadratic case for simplicity,
  \Cref{thm:quad_reg} below can be shown to generalize to a polynomial
  specification of any fixed order.}
Assume for analytical simplicity that $X_t$ has a standard normal distribution,
so in particular $X_t$ and $X_t^2$ are uncorrelated (our qualitative conclusions
can be shown to go through without the normality assumption). Then the
population version of the projection is
\begin{equation*}
  Y_{t+h} = \beta_{0,h} + \beta_{1,h} X_t + \beta_{2,h} X_t^2 + \text{residual}_{h, t+h},
\end{equation*}
with implied derivative of the regression function at $X_t=x$ given by
\begin{equation} \label{eqn:quad_reg_deriv}
\bar{\beta}_h(x) \equiv \beta_{1,h} + 2\beta_{2,h} x,
\end{equation}
and population regression coefficients
\begin{equation} \label{eqn:quad_reg_coef} \beta_{1,h} \equiv
  \frac{\cov(g_{h}(X_{t}), X_t)}{\var(X_t)}, \quad \beta_{2,h} \equiv
  \frac{\cov(g_{h}(X_{t}), X_t^2)}{\var(X_t^2)}.
\end{equation}

\begin{prop}\label{thm:quad_reg}
  Assume that $X_t \sim N(0,1)$, and that $g_{h}$ defined
  in~\eqref{eqn:cond_expn} is differentiable with a derivative that is locally
  absolutely continuous on $\mathbb{R}$. Finally, assume
  $E[\abs{g_{h}(X_{t})}+\abs{g_{h}'(X_{t})}+\abs{g_{h}''(X_{t})}]<\infty$. Then, using the
  definitions~\eqref{eqn:quad_reg_deriv}--\eqref{eqn:quad_reg_coef},
  \begin{equation}\label{eq:bar_beta}
    \bar{\beta}_h(x) = E[(1+X_{t}x)g_h'(X_{t})] = E[g_h'(X_{t})] + xE[g_h''(X_{t})].
  \end{equation}
\end{prop}
The first expression in~\eqref{eq:bar_beta} shows that the estimated derivative
$\bar{\beta}_h(x)$ equals a weighted average of the true derivative function
$g_h'(\cdot)$, but with weights that are negative whenever
$1+X_{t}x<0$.\footnote{It also follows from the proposition that any estimated
  weighted average derivative $\int \omega(x) \bar{\beta}(x)\, dx$ that is a
  nontrivial function of the coefficient $\beta_{2,h}$ (i.e., whenever
  $\int x\omega(x)\, dx \neq 0$) equals a weighted average of $g_h'(\cdot)$ with
  weights that are negative for some $x$.} If the true regression function $g_h$
is in fact quadratic, then $\bar{\beta}_{h}(x)$ is consistent for the
marginal effect function $\Psi_h'(x)$. But if the regression function is
misspecified, the negative weighting leads to a sign reversal: the second
expression in~\eqref{eq:bar_beta} implies that even if $g_h$ is monotonically
increasing, the estimated derivative $\bar{\beta}_h(x)$ will be negative for
sufficiently large $x$ whenever $E[g_h''(X)]<0$. Such sign reversal is not
shared by the linear estimator~\eqref{eqn:lp_reg}, for which the weighting
scheme $\omega_{X}$ is convex. This lack of robustness of a quadratic (or more
generally polynomial) specification of the regression function to functional
form misspecification is related to the observation in \citet{white_using_1980}
that polynomial approximations to the conditional mean function $g_{h}$ cannot
be interpreted as providing a Taylor series approximation to $g_{h}$.

\paragraph{State-dependent specifications.}
A particularly popular nonlinear local projection specification in applied work is a state-dependent specification that interacts the shock with a binary regime indicator $S_t \in \lbrace 0,1 \rbrace$ (see \citealp{Cloyne2023}, \citealp{Goncalves2024}, and references therein):
\begin{equation*}
  Y_{t+h} = \hat{\beta}_{0,h}(1-S_t) X_t + \hat{\beta}_{1,h} S_t X_t +
  \hat{\gamma}_{0,h}'(1-S_t)\vec{W}_t +
  \hat{\gamma}_{1,h}'(S_t\vec{W}_t) + \text{residual}_{h, t+h}.
\end{equation*}
For example, $S_t$ may indicate whether the economy is in a recession. Assuming
that the local projection is fully interacted as above (i.e., all control
variables $\vec{W}_t$ are interacted with $S_t$), then the procedure is
tantamount to running separate regressions on the subsamples with $S_t=0$ and
$S_t=1$, respectively.\footnote{If we instead omit the interaction terms from
  the regression and only control linearly for $S_t$, then we are in the case of
  \Cref{sec:ident-with-contr,sec:ident-with-contr-1} below.} It follows that all
the analysis surrounding \Cref{thm:obsshock} above applies upon conditioning on
$S_t=s \in \lbrace 0,1\rbrace$. In particular, the probability limit of the
state-dependent impulse response estimate $\hat{\beta}_{s, h}$ equals a
positively weighted average of conditional marginal effects
$\partial E[Y_{t+h} \mid X_t=x,S_t=s]/\partial x$, which have a clear causal
interpretation provided the shock independence assumption~\eqref{eqn:indep}
holds conditional on $S_t$ (i.e., within each regime). Thus, despite their
apparent linearity conditional on regime, state-dependent local projections
identify causal estimands even when the true \ac{DGP} has a nonlinear form, such
as a model with smooth or discrete regime-switching. However, consistent with
the discussion in \Cref{sec:framework}, it is important to interpret the impulse
responses as averaging over all future shocks, \emph{including} potential future
regime switches. In other words, the local projection estimand does not hold the
regime fixed within the impulse response horizon.

\subsection{Identification with proxies}\label{sec:proxy}

In many applications, observations of the shock are contaminated by measurement
error, such as when accurate measurements are available only in a subset of the
time periods. In such cases, researchers typically treat the measurements
$Z_{t}$ as a \emph{proxy} for the shock of interest $X_{t}$, or, equivalently,
an instrument for the shock \citep[see][for a review]{Stock2018}. We now show
that when the structural function is nonlinear, linear \acp{VAR} and local projections onto
the proxy identify average marginal effects up to scale, provided that the
conditional mean of the proxy given the shock is monotone in the shock.

We assume that the proxy $Z_{t}$ is valid, in the sense that it satisfies the
exclusion restriction
\begin{equation}\label{eqn:proxy}
E[Y_{t+h} \mid X_t, Z_t] = E[Y_{t+h} \mid X_t] = g_{h}(X_{t}),
\end{equation}
formalizing the notion that if the shock $X_{t}$ were observed, the proxy $Z_{t}$
would not provide any further predictive power. It is implied by the standard
assumption in the measurement error literature that the measurement error in
$Z_{t}$ is non-differential, i.e., the whole conditional distribution of
$Y_{t+h}$ given $(X_{t}, Z_{t})$ depends only on $X_{t}$ (or equivalently that
$Z_{t}$ is independent of $\vec{U}_{h, t+h}$) \citep[e.g.,][Chapter
2.6]{CaRuSt06}.

We consider the ``reduced-form'' local projection of the outcome $Y_{t+h}$ on
the proxy $Z_t$. Under~\eqref{eqn:proxy}, the population version of this
regression has slope coefficient
\begin{equation}\label{eqn:proxy_lp_estimand}
  \tilde{\beta}_h \equiv \frac{\cov(\zeta(X_{t}), g_{h}(X_{t}))}{\var(Z_{t})},
\end{equation}
where
\begin{equation} \label{eqn:proxy_1st}
\zeta(x) \equiv E[Z_t \mid X_t=x].
\end{equation}
As shown by \citet{PMW2021}, $\tilde{\beta}_{h}$ also corresponds to the
probability limit of an impulse response from a structural \ac{VAR} where the
proxy is ordered first, and the specification controls for sufficiently many
lags.

\begin{prop}\label{thm:proxy}
  Assume that $X_t$ is continuously distributed on an interval
  $I\subseteq\mathbb{R}$ (the interval may be unbounded, and could equal
  $\mathbb{R}$), and that the variance of $Z_{t}$ is positive and finite. Assume
  that the conditional mean $g_h$ defined in~\eqref{eqn:cond_expn} is locally
  absolutely continuous on $I$, and
  $E[\abs{g_h(X_{t})}(1+\abs{\zeta(X_{t})})]<\infty$. Finally, assume that
  $\int_{I} \abs{\tilde{\omega}_{Z}(x)g_h'(x)}\, dx<\infty$, where
  \begin{equation} \label{eqn:proxy_weights} \tilde{\omega}_Z(x) \equiv
      \frac{\cov(\1{X_t \geq x}, \zeta(X_t))}{\var(Z_t)},
    \end{equation}
    and that for sufficiently large positive and negative $x$, the function
    $\zeta(x)-E[Z_{t}]$ does not change sign.\footnote{That is, there exist $\underline{x}, \overline{x} \in I$ and $\underline{\iota}, \overline{\iota} \in \lbrace -1,1\rbrace$ such that $\underline{\iota}\lbrace \zeta(x)-E[Z_t]\rbrace \geq 0$ for all $x \leq \underline{x}$ and $\overline{\iota}\lbrace \zeta(x)-E[Z_t]\rbrace \geq 0$ for all $x \geq \overline{x}$.} Then, the proxy
    estimand~\eqref{eqn:proxy_lp_estimand} satisfies
  \begin{equation*}
    \tilde{\beta}_h = \int_I \tilde{\omega}_Z(x)g_h'(x)\, dx.
  \end{equation*}
  The weight function $\tilde{\omega}_Z$ has the following properties:
\begin{enumerate}[(i)]
\item It is equivariant to additive and multiplicative measurement error: If
  $\tilde{Z}_t = V_{1t} + V_{2t}Z_t$, where $(V_{1t}, V_{2t})$ is a bivariate
  random vector independent of $(X_t, Z_t)$, then
  $\tilde{\omega}_{\tilde{Z}}(x) =
  \frac{E[V_{2t}]\var(Z_{t})}{\var(\tilde{Z}_{t})}\tilde{\omega}_Z(x)$ for all
  $x$.
\item\label{itm:proxy_weights_pos} It is nonnegative,
  $\tilde{\omega}_Z(x) \geq 0$, provided that
  $E[Z_{t} \mid X_t \geq x] \geq E[Z_t \mid X_t < x]$.
\item It depends only on the joint distribution of $(X_t, Z_t)$, but
  not on the conditional distribution of $Y_{t+h}$ given $X_{t}$.
\end{enumerate}
A sufficient condition for property (\ref{itm:proxy_weights_pos}) is that the
conditional mean function $\zeta(x)$ is monotone increasing. Under this
assumption, $\tilde{\omega}_Z$ is also hump-shaped: monotonically increasing
from 0 to its maximum for $x \leq x_{0}$, and then monotonically decreasing back
to 0 for $x \geq x_{0}$, where
$x_{0}\equiv\inf\{x\in I\colon \zeta(x)\geq E[Z_{t}]\}$.
\end{prop}

Combining \Cref{thm:proxy} with the identification
result~\eqref{eqn:obsshock_identif} implies that linear proxy regressions
identify weighted averages of marginal effects,
$\theta_h(\tilde{\omega}_Z) = \int \tilde{\omega}_Z(x)\Psi_h'(x)\, dx$, just as
in the case of directly observed shocks. Unlike in the observed shocks case, the weights $\tilde{\omega}_{Z}$ will not be
positive unless the proxy satisfies the condition in
point~(\ref{itm:proxy_weights_pos}) of \Cref{thm:proxy}---this condition is
slightly weaker than monotonicity of $\zeta(x)$.\footnote{For instance, the
  condition may still hold even if monotonicity is violated over a sufficiently
  small interval in the middle of the support of $X_{t}$.} However, monotonicity
of $\zeta$ ensures not just that the weights are positive, but also that they
have an intuitive hump-shape, giving most weight to shocks in the middle of the
distribution.

Monotonicity of $\zeta$ is implied by, but much weaker than the
continuous-treatment version of the \citet{ImAn94} monotonicity condition,
needed for causal interpretation of two-stage least squares estimands under
endogeneity. We defer the details to \Cref*{app:ident-with-instr}, where we
generalize the identification results in \citet{AnGrIm00} by allowing for
non-smooth potential outcome functions and non-binary $Z_{t}$; we also extend
\Cref{thm:proxy} to the case with covariates. It follows from this identification
result that monotonicity of $\zeta$ holds under much weaker conditions than
those required for causal interpretation of $\tilde{\beta}_{h}$ under
endogeneity. \Citet[Theorem 7]{Rambachan2021} derive an alternative
characterization of the proxy estimand~\eqref{eqn:proxy_lp_estimand} involving
derivatives of the \emph{reduced-form} potential outcome as a function of the
proxy $Z_t$ (rather than of the shock $X_t$), and therefore the monotonicity
assumption has no counterpart in their analysis.

A practical implication of \Cref{thm:proxy} is that applied researchers should
seek to construct proxies that are credibly positively related to the unobserved
latent shock of interest. However, it is not essential that the relationship is
linear or indeed of any particular known functional form. Of course,
constructing valid proxies may be challenging in practice: measurement error in their construction may fail to be non-differential, leading
to violation of the exclusion restriction~\eqref{eqn:proxy}, and lack of
variability of $Z_{t}$ in the data may lead to large standard errors for
$\tilde{\beta}_{h}$. The purpose of \Cref{thm:proxy} is to clarify the value
of this hard work, if done well.

\begin{exm}
An interesting example of a proxy is one constructed from so-called ``narrative sign restrictions'', where it is assumed that the researcher observes not the shock itself, but a discrete signal of whether a large shock occurred. While \citet{AntolinDiaz2018} and \citet{Giacomini2023} exploit such restrictions in a likelihood framework, \citet{PMW2021} and \citet{PM2022} recommend treating them as a special case of proxy identification.

As a concrete example, assume that for some constants $c_{1}, c_{2}\geq 0$ (which
may be unknown to the econometrician),
$Z_t = \1{X_t \geq c_2} - \1{X_t \leq -c_1}$. That is, the
proxy equals 1 for sufficiently large positive shocks, $-1$ for sufficiently
large negative shocks, and is otherwise uninformative.\footnote{This example
  assumes that we correctly classify \emph{all} episodes with shocks of
  sufficiently large magnitude. However, \Cref{thm:proxy} shows that the
  calculations continue to apply (up to scale) even if there is random
  misclassification of the form
  $Z_t=V_t[\1{X_t \geq c_2} - \1{X_t \leq -c_1}]$, where $V_t$
  is a Bernoulli random variable that is independent of $(X_t, Y_{t+h})$.} Let
$F_X(x) \equiv P(X_t \leq x)$ be the \ac{CDF} of $X_t$. Then the weight function
$\tilde{w}_Z(x)$ is nonnegative and proportional to
\begin{equation*}
  \cov(\1{X_t \geq x}, Z_t) = \begin{cases}
    F_X(x)[2-F_X(c_2)-F_X(-c_1)] & \text{for } x \leq -c_1, \\
    F_X(x)[1-F_X(c_2)-F_X(-c_1)] + F_X(-c_1) & \text{for } x \in (-c_1,c_2), \\
    [1-F_X(x)][F_X(c_2)+F_X(-c_1)] & \text{for } x \geq c_2,
  \end{cases}
\end{equation*}
as can be verified through direct calculation. It is easy to see that the above weight function is hump-shaped: monotonically increasing until either $x=-c_1$ or $x=c_2$ (depending on the sign of $1-F_X(c_2)-F_X(-c_1)$), and then monotonically decreasing. Arguably, such a weight function is economically sensible. In fact, if $1-F_X(c_2)=F_X(-c_1)$ (as would be the case if $c_1=c_2$ and the distribution of $X_t$ were symmetric around 0), then the weight function is ``nearly'' uniform as it is shaped like a plateau: increasing for $x<-c_1$, then flat for $x \in [-c_1,c_2]$, then decreasing.

This example shows that conventional proxy local projections or \acp{VAR} can estimate meaningful causal summaries even if the proxy (which here is discrete) is quite nonlinearly related to the true (continuous) shock, and in ways that are not directly known to the econometrician. This robustness may not be shared by likelihood-based approaches to identification via narrative restrictions.
\end{exm}

The weight function~\eqref{eqn:proxy_weights} does not integrate to 1 due to
attenuation bias, so that we only identify average marginal effects up to scale.
However, since the weight function doesn't depend on the outcome, this is not an
issue in practice: we can scale $\tilde{\beta}_{h}$ by the response of some
normalization variable to the proxy (this is the so-called unit effect
normalization) to identify a \emph{relative} marginal effect. The local
projection instrumental variable estimator of \citet{Stock2018}, which is a
two-stage least squares version of local projection, automatically performs this
normalization.

Since the shock $X_t$ is not directly observed, we cannot generally estimate the
weight function $\tilde{\omega}_Z$ in the data. Instead, it may be useful to
plot the observed-shock weight function~\eqref{eqn:obsshock_weights} pretending
that $Z_{t}$ is the actual shock of interest. If it happens that
$Z_t \approx X_t$, then these weights will be close to the proxy weights
$\tilde{\omega}_Z$, so the plot provides a ``best-case'' scenario.

\subsection{Identification with control variables}\label{sec:ident-with-contr}
In applications where it is challenging to isolate purely exogenous shifts in policy or fundamentals, researchers may be willing to assume that the observed variable $X_t$ (which could be a policy instrument) is exogenous conditional on some control variables $\vec{W}_t$ (such as variables that comprise the policy-makers information set):
\begin{equation}\label{eqn:selection_obs}
X_t \independent \vec{U}_{h, t+h} \mid \vec{W}_t.
\end{equation}
This is a selection on observables assumption as in \citet{Angrist2011} and \citet{Angrist2018}. For example, the assumption holds if $X_t = \Upsilon(\varepsilon_t, \vec{W}_t)$, where $\varepsilon_t$ is a shock that is independent of $(\vec{W}_t', \vec{U}_{h, t+h}')'$, a nonparametric version of the recursive (or Cholesky) assumption in linear structural \ac{VAR} identification \citep[e.g.,][]{Christiano1999}. Then the conditional expectation function
\begin{equation*}
  g_h(x, \vec{w}) \equiv E[Y_{t+h} \mid X_t=x, \vec{W}_t=\vec{w}]
\end{equation*}
equals the \emph{conditional} average structural function in the causal
model~\eqref{eqn:causal}:
\begin{equation*}
  g_h(x, \vec{w}) = E[\psi_h(x, \vec{U}_{h, t+h}) \mid
  \vec{W}_t=\vec{w}] \equiv \Psi_h(x, \vec{w}),
\end{equation*}
where the expectation is taken with respect to the conditional distribution of
the nuisance shocks $\vec{U}_{h, t+h}$ given $\vec{W}_{t}$.

Even under the selection on observables assumption~\eqref{eqn:selection_obs},
the local projection with controls~\eqref{eqn:lp_reg} need not estimate an
average marginal effect if the relationship between $X_t$ and the controls is
nonlinear. This result extends to recursively identified structural \acp{VAR},
due to the nonparametric equivalence between these procedures \citep{PMW2021}.
\Cref{theorem:ols_covariates} below shows that the population local projection
coefficient $\beta_h$ can still be written as a weighted average of the marginal
effects
$\partial g(x, \vec{w})/\partial x=\partial \Psi_h(x, \vec{w})/\partial x$, but
the weights can be negative if the true ``propensity score''
$\pi^*(\vec{w}) \equiv E[X_t \mid \vec{W}_t = \vec{w}]$ is nonlinear (or, in other words, if the purported ``shock'' that is the residual from a linear projection $X_t$ on the controls $\vec{W}_t$ is predictable by \emph{nonlinear} transformations of the controls). We leave
the details, which extend the analysis of \citet{GoldsmithPinkham2024} to cases
with non-discrete $X_t$, to \Cref{sec:mte}. As usual, negative weights are
worrying, as they may lead to a sign-reversal. If $\vec{W}_{t}$ just consists of a set of mutually exclusive dummies,
then linearity of the propensity score comes for free, and the weights are
guaranteed to be positive. However, outside this case, we recommend that researchers do careful
sensitivity checks with respect to both the set of controls and the functional
form for the controls (e.g., whether the addition of nonlinear transformations
of the controls, such as interactions and polynomials, improves prediction of
$X_{t}$). More research is warranted on best practices for such sensitivity checks, given the data limitations in applied macroeconomics.

\section{The bad: identification via heteroskedasticity}\label{sec:heterosk}

Identification via heteroskedasticity has become a popular procedure for causal
identification in applications where direct shock measures are unavailable,
following \citet{Sentana2001}, \citet{Rigobon2003}, and
\citet{Rigobon2004}.\footnote{Similar identification approaches were developed
  in the signal-processing literature in the 1990s, see the review by
  \citet[Section 18.2]{Hyvarinen2001}.} In a pair of highly-cited papers,
\citet{Lewbel2012,Lewbel2018} exploits this idea to achieve identification in
cross-sectional regressions with endogenous variables and no external
instruments \citep[see also][for a related approach]{Klein2010}. In stark contrast to \Cref{sec:shock}, the results
in this section deliver \emph{bad} news regarding the sensitivity of
identification approaches via heteroskedasticity to the assumption that the
underlying structural function is linear: the Rigobon-Sack-Lewbel estimator does
not generally estimate average marginal effects; more generally, we show that
the nonparametric analogue of the identification approach yields very large
identified sets for causal effects. However, on the positive side, it is
possible to test the linearity assumption in the data.

\subsection{Nonparametric version of the identification approach}

To explain the sensitivity of conventional identification via
heteroskedasticity to the linearity assumption on the structural function, it
is helpful to first lay out a nonparametric version of the framework before we
review the linear case. Since this section is mainly concerned with giving
examples of how the identification approach can fail, we specialize the dynamic
set-up from \Cref{sec:framework} to a simpler static model.

\paragraph{Nonparametric setup.}

We observe an $n$-dimensional vector $\vec{Y}$ of variables that are
nonlinearly related to a latent, scalar shock of interest $X$ as well as an
$(m-1)$-dimensional latent vector $\vec{U}$ of nuisance shocks:
\begin{equation} \label{eqn:factor_nonparam} \vec{Y} =
  \vec{\psi}(X, \vec{U}), \quad X \independent \vec{U},
\end{equation}
where we suppress time subscripts to ease notation. The above model is a
(static) \emph{nonparametric factor model}, since we do not impose parametric
restrictions on the unknown structural function
$\vec{\psi} \colon \mathbb{R}^m \to \mathbb{R}^n$.

The econometrician observes a scalar $D$ that is informative about the
heteroskedasticity of the shock of interest $X$ but independent of the nuisance
shocks $\vec{U}$ (jointly with $X$):
\begin{equation}\label{eqn:hetero_indep}
(D, X) \independent \vec{U}.
\end{equation}
This assumption implies that $D$ is a valid proxy for $X$, in the sense that
$\vec{Y}$ and $D$ are independent conditional on $X$. But because the
variable $D$ only influences the variance and higher moments of $X$ but not its
mean,
\begin{equation} \label{eqn:hetero_mean}
E[X \mid D]=0,
\end{equation}
we cannot use the proxy in local projections as in \Cref{sec:proxy}. For
concreteness, it may be useful to think of $D$ as a binary \emph{regime
  indicator}, which affects the conditional variance $\var(X \mid D)$ but not
the conditional mean~\eqref{eqn:hetero_mean}, as in the original work by
\citet{Rigobon2003}.

If we assume that the structural function $\vec{\psi}$ is linear, it is
possible to achieve identification even if $D$ is unobserved, and we
relax~\eqref{eqn:hetero_indep} by allowing $D$ to affect the variances of nuisance shocks. See
\citet{Bacchiocchi2024} and \citet[Section 3]{Lewis2024} for excellent reviews. However, since we are only
interested in showing how the basic identification approach can fail in a
nonparametric context, we maintain the stronger assumptions above. It then
follows \emph{a fortiori} that nonparametric identification is even more
challenging under weaker assumptions.

\paragraph{Review of linear identification.}
If the structural function $\vec{\psi}$ in~\eqref{eqn:factor_nonparam} is known to be partially linear, identification of causal effects obtains under an additional relevance assumption. Thus, we temporarily assume that
\begin{equation} \label{eqn:hetero_linear}
\vec{\psi}(x, \vec{u})= \vec{\theta} x + \vec{\gamma}(\vec{u}),
\end{equation}
where $\vec{\theta}$ is the unknown vector of causal effects of $X$, while $\vec{\gamma} \colon \mathbb{R}^{m-1} \to \mathbb{R}^n$ is an unknown function. Following \citet{Rigobon2004} and \citet{Lewbel2012}, construct the scalar instrumental variable
\begin{equation} \label{eqn:hetero_iv}
Z \equiv (D-E[D])Y_1,
\end{equation}
where $Y_1$ is the first element of $\vec{Y}$. In applications, $Y_1$ may be a policy instrument that is known to be strongly related to $X$, though it is also allowed to be correlated with the nuisance shocks. Under the linear model~\eqref{eqn:hetero_linear} and the identification assumptions~\eqref{eqn:hetero_indep}--\eqref{eqn:hetero_mean}, $Z$ satisfies the exogeneity restriction for linear identification in \citet{Stock2018} since $E[Z \mid \vec{U}]=0$. In particular, under these assumptions, a regression of $\vec{Y}$ on $Y_1$ using $Z$ as an instrument identifies the (relative) causal effects of $X$:
\begin{equation}\label{eqn:hetero_identif_linear}
\frac{1}{\cov(Y_1,Z)}\cov(\vec{Y}, Z) = \frac{1}{\theta_1}\vec{\theta}.
\end{equation}
To ensure we are not dividing by zero, we need to additionally assume the relevance conditions that (i) the shock of interest is heteroskedastic across regimes, $\cov(X^2,D) \neq 0$, and (ii) the causal effect of $X$ on $Y_1$ is nonzero, $\theta_1 \neq 0$. For completeness, we review the calculations leading to~\eqref{eqn:hetero_identif_linear} in \Cref*{app:hetero_identif_linear}.

\subsection{Fragility under nonlinearity}
We now argue that the simple linear identification argument fundamentally cannot
be extended to nonparametric contexts.

\paragraph{Nonparametric identified set.}
We first show that the nonparametric model of identification via
heteroskedasticity yields a large identified set for the causal effects of $X$
on $\vec{Y}$. To do this, we strengthen the independence and conditional mean
assumptions~\eqref{eqn:hetero_indep}--\eqref{eqn:hetero_mean} by imposing a
specific model for the relationship between $X$ and $D$:
\begin{equation}\label{eqn:hetero_X}
X = \sigma(D)\Ra, \quad \text{where $\Ra, D$, and $\vec{U}$ are mutually independent,}
\end{equation}
$\sigma \colon \mathbb{R} \to \mathbb{R}_+$ is a \emph{known} function, and $\Ra$ has a \emph{known} distribution that is symmetric around 0. This model would, for example, be consistent with the conditionally Gaussian model $X \mid D \sim N(0,\sigma^2(D))$.

\begin{prop}\label{thm:hetero_identif}
  Assume that $(\vec{Y}, D, \Ra, X, \vec{U})$ satisfy the nonparametric
  factor model~\eqref{eqn:factor_nonparam} and identification
  assumption~\eqref{eqn:hetero_X}. Then there exists an alternative structural
  function $\tilde{\vec{\psi}} \colon \mathbb{R}^2 \to \mathbb{R}^n$ and
  a scalar random variable $\tilde{U}$ independent of $(\Ra, D, X)$ such that
  $(\tilde{\vec{Y}}, D)$ has the same joint distribution as
  $(\vec{Y}, D)$, where
  \begin{equation*}
    \tilde{\vec{Y}} \equiv \tilde{\vec{\psi}}(X, \tilde{U}),
  \end{equation*}
and such that $\tilde{\vec{\psi}}(-x, \tilde{u})=\tilde{\vec{\psi}}(x, \tilde{u})$ for all $x, \tilde{u}$.
\end{prop}

The proposition states that the identified set for $\vec{\psi}$ is so large that
it always contains a structural function $\vec{\psi}(x, \vec{u})$ that is
symmetric in $x$ around 0. In particular, we can never rule out that the average
marginal effect
$\int \omega(x) (\partial E[\vec{\psi}(x, \vec{U})]/\partial x)\, dx$ is zero
when the weight function $\omega(x)$ is symmetric around 0. Intuitively, the
challenge is that $D$ does not affect the mean of $X$, only higher moments,
so---without strong functional form restrictions on the relationship between the
outcomes and the shocks---we do not have enough information to sign mean effects
of shifts in the latent shock $X$. This holds even though we assume that the
econometrician knows exactly how $D$ affects the dispersion of the $X$
distribution. Notice that the construction of the observationally equivalent
symmetric structural function in \Cref{thm:hetero_identif} only relies on a
single (scalar) nuisance shock; hence, knowledge about the true number of shocks
does not ameliorate the identification failure (see \Cref{sec:nongauss} for
further discussion of this point).

A careful inspection of the proof of \Cref{thm:hetero_identif} reveals that the
result is closely related to a known issue with identification via
heteroskedasticity in a linear context: while the variance of the shock of
interest $X$ must vary across regimes, we cannot simultaneously allow the
impulse responses of $X$ to vary across regimes \citep[see][Section 6.1, for a
discussion and references]{Lewis2024}. However, in a nonparametric context this
problem is even worse, since there is no fundamental distinction between
``coefficients'' and ``shock variances'' in a general nonlinear model.
\emph{A priori}
restrictions that certain ``coefficients'' are independent of the regime are
only meaningful once we parametrize the model, which complicates the development of an
empirically useful nonparametric generalization of the identification approach.

\paragraph{Sensitivity of linear procedures.}

Because the nonparametric identified set is large, we can expect estimation
procedures based on linearity of the structural function to fail to estimate
causal objects in general. The next result implies that this is indeed the case
for the linear instrumental variable estimator~\eqref{eqn:hetero_identif_linear} of \citet{Rigobon2004} and
\citet{Lewbel2012}.

\begin{prop}\label{thm:hetero_iv}
Assume the additively separable structural model
\begin{equation*}
  \vec{Y} = \vec{\Psi}(X) + \vec{\gamma}(\vec{U}),
\end{equation*}
where $\vec{\Psi} \colon \mathbb{R} \to \mathbb{R}^n$,
$\vec{\gamma} \colon \mathbb{R}^{m-1} \to \mathbb{R}^n$, and we normalize
$E[\vec{\Psi}(X)]=E[\vec{\gamma}(\vec{U})]=\vec{0}$. Suppose that the
independence assumption~\eqref{eqn:hetero_indep} holds, and let $Z$ be given
by~\eqref{eqn:hetero_iv}. Suppose that the variables $(\vec{Y}, Z, D)$ have
finite second moments, and that the support of $X$ is given by the interval
$I\subseteq\mathbb{R}$ (the interval may be unbounded, and could equal
$\mathbb{R}$). Suppose also that for each $j$, $\Psi_{j}$ (the $j$-th component
of $\vec{\Psi}$) is locally absolutely continuous on $I$, and that for some
$\underline{x}, \overline{x}\in I$, $\Psi_{j}(x)$ is monotone for
$x\leq \underline{x}$ and for $x\geq \overline{x}$. Then
\begin{equation*}
  \cov(\vec{Y}, Z) = \int \check{\omega}(x) \vec{\Psi}'(x)\, dx,
\end{equation*}
where
\begin{equation}\label{eqn:hetero_iv_weights}
\check{\omega}(x) \equiv \cov\big(\1{X \geq x}, \Psi_1(X)(D-E[D])\big).
\end{equation}
\end{prop}

\Cref{thm:hetero_iv} shows that regressing $\vec{Y}$ onto the instrument
$Z$ yields a weighted average of marginal effects, but with a weight function
$\check{\omega}(x)$ that cannot be guaranteed to be positive.\footnote{This is
  not a special case of \Cref{thm:proxy}, since $Z$ does not satisfy the
  nonparametric proxy assumption~\eqref{eqn:proxy}.} In fact, the weights even
integrate to 0 in some cases, for example if $\Psi_1(x)=\Psi_1(-x)$ and the
conditional distribution of $X$ given $D$ is symmetric around 0. In such cases,
the instrument erroneously estimates a zero causal effect
of $X$ on $Y_j$ for $j \geq 2$ even if $\Psi_j(x)=\theta_j x$ is a linear
function with $\theta_j \neq 0$.

The weights can also be negative---and therefore cause the econometrician to get
the sign of the marginal effects wrong---even in the seemingly favorable setting
where (unbeknownst to the econometrician) the policy variable $Y_{1}$ simply
equals the shock of interest $X$, without any nonlinearity or contamination by
nuisance shocks, i.e., $\Psi_{1}(x)=x$. Assume in addition, as in
\citet{Rigobon2003}, that the regime indicator $D \in \lbrace 0,1\rbrace$ is
binary and $E[X \mid D]=0$. Then a simple calculation shows that the
weights in~\cref{eqn:hetero_iv_weights} equal
\begin{equation}\label{eqn:hetero_iv_weights_special}
\check{\omega}(x) = \var(D) \int_x^{\infty} [f_{X|1}(v)-f_{X|0}(v)]v\, dv = \var(D) \int_{-\infty}^x [f_{X|0}(v)-f_{X|1}(v)]v\, dv,
\end{equation}
where $f_{X|d}(x)$ is the density of $X$ conditional on regime $D=d$. Suppose
the right (resp., left) tail of the $X$ distribution is fatter (resp., thinner)
in regime $D=1$ than in regime $D=0$, meaning that $f_{X|1}(x) > f_{X|0}(x)$
for $x \gg 0$ and $f_{X|0}(x) > f_{X|1}(x)$ for $x \ll 0$. Then it follows
from \cref{eqn:hetero_iv_weights_special} that $\check{\omega}(x)>0$ for
$x \gg 0$, while $\check{\omega}(x)<0$ for $x \ll 0$. This simple example shows
that the instrumental variable estimator can easily generate negative weights,
even when it satisfies the exclusion and relevance conditions and the policy
variable is linear in the shock. To trust that the weights are positive, we
would need to have quite detailed information about the conditional shock
density in the two regimes; simple moment restrictions do not suffice.

Intuitively, the problem of negative weights comes about because the
\citet{Rigobon2003} and \citet{Lewbel2012} instrumental variable $Z$ defined
in~\eqref{eqn:hetero_iv} fails the proxy monotonicity assumption discussed
earlier in connection with \Cref{thm:proxy}. Because the only source of
exogenous variation is the regime indicator $D$, and this indicator does not
affect the mean of the latent shock $X$ but only higher moments, it is generally
impossible to construct any proxy variable that is guaranteed to be monotone in
$X$, unless we make strong assumptions about the structural function.

If the model is not additively separable as assumed in \Cref{thm:hetero_iv}, the
instrumental variables estimator can exhibit even more pathological behavior, in
that it may not equal a weighted average of marginal effects at all. As a simple
example, consider a multiplicative model
$\vec{Y} = X\vec{\gamma}(\vec{U})$ with
$E[\vec{\gamma}(\vec{U})]=\vec{0}$ and impose the independence
assumption~\eqref{eqn:hetero_indep}. In that model, $E[\vec{Y} \mid X] = \vec{0}$,
so the marginal effect function is identically zero, but
$\cov(\vec{Y}, Z)=\cov(X^2,D)\cov(\gamma_1(\vec{U}), \vec{\gamma}(\vec{U}))
\neq 0$ in general, so the instrument erroneously estimates a nonzero effect.

\subsection{Silver lining: testability of the linearity assumption}

While the sensitivity of identification via heteroskedasticity to linearity of
the structural function $\vec{\psi}$ is disheartening, at least the linear
model~\eqref{eqn:hetero_linear} implies testable restrictions. Specifically, as
noted by \citet{Rigobon2004} and \citet{Wright2012}, for any $d_0,d_1$ in the
support of $D$, the difference
$\var(\vec{Y} \mid D=d_1)-\var(\vec{Y} \mid D=d_0) = [\var(X \mid D=d_1)-\var(X
\mid D=d_0)]\vec{\theta}\vec{\theta}'$ should be a rank-1 matrix under linearity
and the maintained independence assumption~\eqref{eqn:hetero_indep}. Other
over-identification tests in more general linear models of identification via
heteroskedasticity are discussed in the review article by \citet{Lewis2024}. We
are not aware of any thorough analysis of the power properties of these tests
against nonlinear alternatives.

\section{The ugly: identification via non-Gaussianity}\label{sec:nongauss}

A second approach to identification in linear models in the absence of direct
shock measures is to assume that the structural shocks are mutually
independent and non-Gaussian; see \citet{Gourieroux2017}, \citet{Lanne2017}, and
the review article by \citet{Lewis2024}. \Citet{Lewbel2024} propose a similar
approach to achieve identification in cross-sectional endogenous regressions in
the absence of external instruments. An earlier literature outside economics
goes by the name \acf{ICA}, see \citet[Chapter 10]{Kagan1973}, \citet{Comon1994},
and the textbook by \citet{Hyvarinen2001}.

This section delivers \emph{ugly} news regarding the sensitivity of this
identification approach to linearity of the structural function: once the
linearity assumption is dropped, the non-Gaussianity assumption is essentially
vacuous; as a consequence, estimators based on non-Gaussianity and linearity of
the structural function can fail spectacularly even under mild departures from
linearity. What is worse, the linearity assumption is untestable in general.

\subsection{Nonparametric version of the identification approach}
As in \Cref{sec:heterosk}, we consider the nonparametric factor
model~\eqref{eqn:factor_nonparam}. However, unlike in the case of identification
via heteroskedasticity, we now do not observe any additional proxy variables
that aid in identifying the latent shocks. Instead, we hope to achieve
identification via restrictions on the distributions of the shocks.

\paragraph{Review of linear identification.}
Assume temporarily that the number of shocks equals the number of observables,
$m=n$, and that the structural function is linear:
\begin{equation*}
  \vec{\psi}(x, \vec{u}) = \vec{\theta}x
  + \vec{\gamma}\vec{u}, \quad \vec{\theta} \in \mathbb{R}^n,\;
  \vec{\gamma} \in \mathbb{R}^{n \times (n-1)}.
\end{equation*}
Assume also that the $n$ shocks $(X, U_1, \dotsc, U_{n-1})$ are mutually independent,
and at most one of these shocks has a Gaussian distribution. We will refer to
this model as the linear \ac{ICA} model. A deep result in probability theory,
the Darmois-Skitovich Theorem, says that two nontrivial linear combinations of
independent variables cannot themselves be independent, unless all the
underlying variables are Gaussian. In the context of the linear \ac{ICA} model,
the theorem implies that any two linear combinations $\vec{\varsigma}'\vec{Y}$
and $\tilde{\vec{\varsigma}}'\vec{Y}$ of the data
$\vec{Y} = \vec{\theta}X + \vec{\gamma}\vec{U}$ can be independent if and only
if these linear combinations equal two different shocks in the model (up to sign
and scale). Hence, the shocks in the model can be identified by searching for
those linear combinations of the observed variables that are independent; once
we have the shocks, we can then estimate their causal effects. See
\citet{Hyvarinen2001} and \citet[Section 4]{Lewis2024} for reviews of estimation
procedures.

The abstract identification argument above can be made less mysterious through a
method of moments framework that exploits implications of shock independence for
higher moments of the data \citep[Section 4.4]{Lewis2024}. Nevertheless, it is
clear from both the abstract argument and the more concrete moment-based
approach that linearity of the structural function is being leveraged heavily.

\subsection{Fragility under nonlinearity}

\paragraph{Nonparametric identified set.}
Unfortunately, the mere assumptions that the latent shocks are independent and non-Gaussian provide essentially no identification power in a nonparametric context. The identified set under these assumptions is so large that nearly any function of the data can be labeled a ``shock''.

\begin{prop}\label{thm:factor}
  Let $\tilde{\vec{Y}} = \vec{\Upsilon}(\vec{Y})$ be a
  homeomorphic\footnote{That is, continuous, one-to-one, and with a continuous
    inverse function $\vec{\Upsilon}^{-1}(\cdot)$.} transformation of $\vec{Y}$,
  with $j$-th element denoted by $\tilde{Y}_j$. For all $j=2,\dotsc, n$, assume
  that the quantile function of $\tilde{Y}_j$ conditional on
  $\tilde{Y}_{j-1}, \tilde{Y}_{j-2}, \dotsc, \tilde{Y}_1$ is continuous in the
  quantile and the conditioning arguments. Define
  $\tilde{X} \equiv \tilde{Y}_1$, and let
  $\lbrace \bar{U}_j \rbrace_{j=1}^{n-1}$ be mutually independent uniform
  variables on $[0,1]$ that are also independent of $\tilde{X}$. Then there
  exists a continuous function
  $\bar{\vec{\psi}} \colon \mathbb{R}^n \to \mathbb{R}^n$ such that the random
  vector
  \begin{equation*}
    \bar{\vec{Y}} \equiv \bar{\vec{\psi}}(\tilde{X}, \bar{U}_1, \dotsc, \bar{U}_{n-1})
  \end{equation*}
has the same distribution as $\vec{Y}$.
\end{prop}
\Cref{thm:factor} shows that the nonparametric factor model~\eqref{eqn:factor_nonparam} is very under-identified, even if we restrict the number of latent, independent shocks to equal $m=n$ and impose smoothness on the structural function $\vec{\psi}$. Indeed, any element of almost any one-to-one transformation $\vec{\Upsilon}(\vec{Y})$ of the observables could be construed as a ``shock'' for \emph{some} data-consistent choice of structural function $\vec{\psi}$. Note that the identification issues go well beyond the familiar ``labeling problem'' in linear ICA analysis, where the shocks are identified up to sign and permutation, such that additional economic information is required to label each of the statistically identified shocks \citep[Section 6.4]{Lewis2024}. While the challenge of identifying nonlinear factor models is well known in the broader literature (see the review by \citealp{Jutten2003}), it appears that the serious consequences of this fact for shock identification in macroeconometrics have not been explored previously.

The fundamental issue is that non-Gaussianity of the shocks is a vacuous
assumption in the nonparametric setting: it is an innocuous normalization to
assume that all shocks have uniform distributions, since we can always
nonlinearly transform any shock distribution to the uniform distribution via the
quantile function. In other words, we have severe identification failure as in
\Cref{thm:factor} even if the econometrician knows the exact distributions of
each shock. Hence, it is no accident that the identification argument in the
\ac{ICA} and structural \ac{VAR} literatures relies heavily on the linearity
assumption: there is no nonlinear equivalent of the Darmois-Skitovich
Theorem.

If we allow for slightly less smoothness of the structural function
$\vec{\psi}$, then the identification problem is even worse. As
\citet{Gunsilius2023} note, any $n$-dimensional vector $\vec{Y}$ can be
represented as a nonlinear factor model~\eqref{eqn:factor_nonparam} in a
\emph{single} latent shock $X$ (so $m=1$ and $\vec{U}=\vec{0}$)
using a so-called space-filling curve (e.g., Hilbert curve) construction, though
the associated $\vec{\psi}$ function would not be one-to-one. Hence,
without restrictions on the structural function, we cannot rule out that the
latent shock of interest $X$ drives \emph{all} the variation in the $n$ observed
variables $\vec{Y}$. Indeed, given a uniform random variable $U$ on $[0,1]$, we
can generate an infinite number of independent uniform random variables
$\{V_{j}\}$ from the decimal expansion of
$U=0.u^{1}_{1}u_{2}^{1}u_{1}^{2}u_{3}^{1}u_{2}^{2}u_{1}^{3}u_{4}^{1}u_{3}^{2}u_{2}^{3}u_{1}^{4}\dotsb$,
and taking $V_{j}\equiv 0.u^{j}_{1}u^{j}_{2}\dotsb$ (and hence an infinite number of
independent random variables with arbitrary distributions $F_{j}$ by taking the
inverse transform $F_{j}^{-1}(V_{j})$).\footnote{This is a consequence of the
  fact that there exists a one-to-one function $\vec{\phi}$ such that both $\vec{\phi}$ and $\vec{\phi}^{-1}$
  are measurable between the measurable spaces $(M, \mathcal{B}_{M})$ and
  $([0,1], \mathcal{B}_{[0,1]})$, where $M$ is any separable complete metric
  space and $\mathcal{B}_{M}$ is the Borel $\sigma$-algebra on $M$
  \citep[Theorem 3.1.1]{dudley02}.}

\paragraph{Sensitivity of linear procedures.}
Due to the nonparametric identification failure, we can expect identification
approaches based on non-Gaussianity to be very sensitive to exact linearity in
the structural function. The following two simple examples illustrate such
sensitivity in two settings where the linearity assumption is untestable. Thus,
identification via non-Gaussianity is not only fragile, it is also generally not
falsifiable.

\begin{exm}
Let the two latent shocks $(X, U)$ have a bivariate standard normal distribution. Let the two observed variables be given by
\begin{equation*}
  Y_1 \equiv X+U, \quad Y_2 \equiv \gamma(X-U),
\end{equation*}
for an arbitrary measurable nonlinear function $\gamma \colon \mathbb{R} \to \mathbb{R}$. We can interpret this setting as being \emph{almost} a linear ICA model, except that the second variable has not been transformed quite correctly. In the above model, $Y_1$ and $Y_2$ are independent, and $Y_2$ has a non-Gaussian distribution.\footnote{This is because the vector $(X+U, X-U)$ has a joint normal distribution with uncorrelated components.} Hence, any linearity-based ICA procedure applied to the data $(Y_1,Y_2)$ will erroneously conclude that the first variable equals the first shock and the second variable the second shock (up to the mean). Moreover, there is nothing in the data that can reject the validity of the linear ICA assumptions. Notice the lack of continuity: even if $\gamma(\cdot)$ is only slightly nonlinear, linear \ac{ICA} procedures will conclude (asymptotically) that the first shock contributes 100\% of the variance of $Y_1$, even though the true number is 50\%.

This example illustrates how getting the transformation of $Y_2$ slightly wrong can mess up causal inference about the other variable $Y_1$ (which is in fact linear in the true shocks).
\end{exm}

\begin{exm}\label{exm:nongauss_counter2}
Consider a model of the form
\begin{equation*}
  Y_1 = X+U, \quad Y_2 = X + \gamma(U),
\end{equation*}
where $X$ and $U$ are independent latent shocks. \Cref*{app:nongauss_counter2}
gives concrete choices of non-Gaussian distributions of the shocks and a smooth
$\gamma(\cdot)$ function such that $Y_1$ and $Y_2$ are independent and both
non-normal. Hence, as in the previous example, any linearity-based ICA procedure
applied to the data $(Y_1,Y_2)$ will erroneously attribute all variation in
$Y_1$ to the first shock and all variation in $Y_2$ to the second shock. Note
that in this example, both of the true shocks $(X, U)$ are non-Gaussian, and the
only nonlinearity in the true structural model is the relationship between
$Y_2$ and $U$.
\end{exm}

In conclusion, linearity-based \ac{ICA} identification procedures can be highly
misleading under departures from a linear model, as with identification via
heteroskedasticity (but unlike identification with observed shocks or proxies).
In fact, the situation is arguably worse than in \Cref{sec:heterosk}, since even
arbitrarily small structural nonlinearities can yield large biases, and the
linearity assumption is not testable in general.

\section{Identification of average marginal effects}\label{sec:mte}

In \Cref{sec:shock}, we reverse-engineered the weight function allowing us to
interpret the linear local projection estimand as an average marginal effect. We
now consider a \emph{forward-engineering} problem: how can we estimate an average marginal effect with a \emph{pre-specified} weight
function? For simplicity, we focus on the case when the shock of interest is
observed, but we drop the requirement that the shock is continuously
distributed. We first consider the case without control variables before
extending the analysis to allow for controls.

\subsection{Identification without controls}\label{sec:ident-with-contr-2}

We consider the setup from \Cref{sec:obsshock}, but drop time subscripts to make
it clearer that our analysis applies to cross-sectional as well as time series
settings. Indeed, part of our goal is to show that identification arguments from
cross-section settings carry directly through to time-series contexts. Let
\begin{equation}\label{eq:conditional_mean}
  g(x)\equiv E[Y\mid X=x]
\end{equation}
denote the conditional mean function from a nonparametric regression of the scalar outcome $Y$
onto the scalar variable $X$. We do not restrict the marginal distribution of $X$: it can be continuous, discrete, or mixed. Let $I\subseteq\mathbb{R}$ denote a (possibly unbounded)
interval that contains the support of $X$. We are
interested in summarizing $g$ by reporting its weighted average derivative,
weighted by some pre-specified weight function $\omega$. With some abuse of
notation, we still denote this weighted average derivative by
\begin{equation*}
  \theta(\omega) \equiv \int_{I}\omega(x)g'(x)\, dx,
\end{equation*}
as in \Cref{sec:framework}, even though we don't require that $g(x)$
corresponds to some structural function. To ensure that this object is
well-defined, we assume that $g$ is locally absolutely continuous on $I$.
Since~\eqref{eq:conditional_mean} only defines $g$ on the support of $X$, this
requires us to extend $g$ to all of $I$ in cases when there are gaps in the
support of $X$, such as when $X$ is discrete. This can be done by linear
interpolation: if $P(X\in (a, b))=0$ for some $(a, b)\subseteq I$, we set
$g(x)=(g(b)-g(a))(x-a)/(b-a)+g(a)$ for $x\in(a, b)$. If the distribution of $X$
is discrete, this defines the derivative $g'$ as the slope between adjacent
support points (and the extension will automatically be locally absolutely
continuous provided that the spacing between adjacent support points is bounded away from 0).

A regression-based approach to estimating $\theta(\omega)$ first estimates the
entire derivative function $g'(\cdot)$ nonparametrically (by, say, series or
kernel regression), and then averages it using the weights $\omega$. The next
result shows that we can alternatively estimate $\theta(\omega)$ as a weighted
average of outcomes, $\theta(\omega)=E[\alpha(X)Y]$, where $\alpha$ is the Riesz
representer of the linear functional $g\mapsto \theta(\omega)$.

\begin{lem}\label{theorem:rr}
  Let $\omega(x) \equiv E[\1{X\geq x}\alpha(X)]$. Suppose that (i) the support of $X$
  is contained in a (possibly unbounded) interval $I\subseteq \mathbb{R}$; (ii)
  $g$ is locally absolutely continuous on $I$; (iii)
  $E[\abs{\alpha(X)}(1+\abs{g(X)})]<\infty$ with $E[\alpha(X)]=0$; and
  (iv) there exists $x_{0}\in I$ such that
  $E[\abs{\alpha(X)\int_{x_{0}}^{X}\abs{g'(x)}\, dx}]<\infty$. Then
  \begin{equation}\label{eq:riesz_representer}
    E[\alpha(X)g(X)]=\int_{I} \omega(x)g'(x)\, dx.
  \end{equation}
\end{lem}

Analogous representations for $\theta(\omega)$ are well-known in the literature if we
additionally assume that $X$ is continuously distributed \citep[e.g.,][equation 2.6]{NeSt93}. The representation is usually derived by
directly applying integration by parts. Our proof instead generalizes the proof
of Stein's lemma \citep[Lemma 1]{stein81}, allowing us to drop the requirement that $X$ is continuously distributed and impose only very mild regularity conditions, which essentially
just require that both sides of \cref{eq:riesz_representer} are
well-defined. In particular, absolute continuity of $g$ is needed to ensure that
$\theta(\omega)$ is well-defined, and in \Cref*{lemma:finite_integral_omega} in
\Cref*{app:proofs}, we show that if the tails of $\omega$ or the tails of $g$ are
monotone, then condition (iv) of \Cref{theorem:rr} holds provided the integral on the right-hand
side of~\eqref{eq:riesz_representer} exists.

\Cref{theorem:rr} gives a recipe for constructing weighting-based estimators of
$\theta(\omega)$ for particular choices of weight function $\omega$ by replacing
the expectation in \cref{eq:riesz_representer} with a sample average and, if the
function $\alpha$ is unknown, replacing $\alpha$ with an estimate.\footnote{As
  we discuss in \Cref{sec:ident-with-contr-1} below, recent results in the
  semiparametric literature suggest that rather than picking between this
  weighting-based approach and the regression-based approach to estimation of
  $\theta(\omega)$, it may pay off to combine them, yielding a ``doubly-robust''
  estimator.} For instance, suppose $X$ is continuous, and let
$\omega(x)=f_{X}(x)$ correspond to the density of $X$, so that
$\theta(\omega)=E[g'(X)]$ is the (unweighted) average derivative. Then the
required weighting is given by $\alpha(x)=-f_{X}'(x)/f_{X}(x)$, leading to the
estimator of \citet{hardle1989} if one uses kernel estimators to estimate the
density and its derivative. If the identification
condition~\eqref{eqn:obsshock_identif} holds, this estimator identifies the
average causal impact of increasing $X$ by an infinitesimal amount. To estimate
the average impact of increasing $X$ by a fixed amount $\delta$ (i.e., the
unweighted average causal effect), which corresponds to setting
$\omega(x)=\frac{1}{\delta}\int_{x-\delta}^{x}f_{X}(x)\, dx$, let
$\alpha(x)=-\frac{f_{X}(x)-f_{X}(x-\delta)}{\delta f_{X}(x)}$, replacing the
derivative of the density by a discrete change. If we set
$\omega(x)=f_{X}^{2}(x)$, so that $\alpha(x)=-2f_{X}'(x)$, and we use a
leave-one-out kernel estimator for the derivative of the density, we recover the
famous density-weighted average derivative estimator of \citet{PoStSt89}.
Finally, for $\omega(x)=E[\1{X\geq x}X]$, we get $\alpha(x)=x$, and
\Cref{theorem:rr} reduces to \Cref{thm:obsshock}: this weighting scheme can be
estimated by linear regression. As the proofs of
\Cref{thm:quad_reg,thm:proxy,thm:hetero_iv} reveal, these results are also
special cases of \Cref{theorem:rr}.\footnote{As a consequence, the assumption
  that $X_{t}$ be continuous in \Cref{thm:obsshock,thm:proxy} can be dropped.}

\subsection{Identification with control variables}\label{sec:ident-with-contr-1}

We now generalize the setup to allow for a vector of controls $\vec{W}$.
Consider the weighted average derivative
\begin{equation*}
\theta(\omega) \equiv E\left[\int_{I_{\vec{W}}} \omega(x, \vec{W})g'(x, \vec{W})\, dx\right],
\end{equation*}
where the expectation is over the marginal distribution of $\vec{W}$, and $g'$
is the derivative with respect to $x$ of the conditional mean function
$g(x, \vec{w}) \equiv E[Y\mid X=x, \vec{W}=\vec{w}]$.\footnote{Weighting by the
  marginal distribution of $\vec{W}$ is not restrictive, since weighting schemes
  that use other forms of averaging across $\vec{w}$ can be recovered by
  defining $\omega$ appropriately.} To ensure this object is well-defined, we
assume that for each $\vec{w}$, the weights $\omega$ are zero outside the
interval $I_{\vec{w}}$ containing the conditional support of $X$ given
$\vec{W}=\vec{w}$, and that we can extend $g(\cdot, \vec{w})$ to $I_{\vec{w}}$
such that $g(\cdot, \vec{w})$ is locally absolutely continuous on $I_{\vec{w}}$,
such as by linearly interpolating across any gaps in the support.

Like in the case without covariates, a regression-based estimator of
$\theta(\omega)$ first estimates the derivative of the regression function
$g'(x, \vec{w})$, and then averages the estimated derivative function
using the weights $\omega$ and the marginal distribution of the covariates.
The next result shows that we can alternatively estimate $\theta(\omega)$ by
taking weighted averages of the outcome.

\begin{lem}\label{theorem:rr-covariates}
  Let
  $\omega(x, \vec{w}) \equiv E[\1{X\geq x}\alpha(X, \vec{W})\mid
  \vec{W}=\vec{w}]$. Suppose that conditional on $\vec{W}$, the following holds
  almost surely: (i) the support of $X$ is contained in a (possibly unbounded)
  interval $I_{\vec{W}}\subseteq \mathbb{R}$; (ii) $g(\cdot, \vec{W})$ is locally
  absolutely continuous on $I_{\vec{W}}$; and (iii)
  $E[\alpha(X, \vec{W})\mid \vec{W}]=0$. Suppose also that (iv) there exists a
  function $x_{0}(\vec{W})\in I_{\vec{W}}$ such that
  $E[\abs{\alpha(X, \vec{W}) \int_{x_{0(\vec{W})}}^{X}
    \abs{g'(x, \vec{W})}\, dx}]<\infty$; and that (v)
  $E[\abs{\alpha(X, \vec{W})}(1+\abs{g(X, \vec{W})})]<\infty$. Then
   \begin{equation}\label{eq:rr-covariates}
    E\left[\int_{I_{\vec{W}}}\omega(x, \vec{W})g'(x, \vec{W})\, dx\right]=
    E\left[\alpha(X, \vec{W})g(X, \vec{W})\right].
  \end{equation}
\end{lem}

As discussed in~\Cref{sec:ident-with-contr-2}, the
representation~\eqref{eq:rr-covariates} is well-known if the distribution of $X$
is continuous conditional on $\vec{W}$. The novelty of
\Cref{theorem:rr-covariates} is to drop the continuity requirement and relax
the regularity conditions.

If $X\in\{0,1\}$ is a binary treatment variable, and we additionally assume that
$X$ is as good as randomly assigned conditional on $\vec{W}$, then
$g(1,\vec{w})-g(0,\vec{w})$ corresponds to the conditional \ac{ATE} for
individuals with $\vec{W}=\vec{w}$. In this case, the average derivative
simplifies to
\begin{align*}
  \theta(\omega)&=E\left[(g(1,\vec{W})-g(0,\vec{W}))
                  \int_{I_{\vec{W}}} \omega(x, \vec{W})\, dx\right] \\
                &=E[(g(1,\vec{W})-g(0,\vec{W}))\alpha(1,\vec{W})P(X=1\mid \vec{W})],
\end{align*}
which corresponds to a weighted average of conditional \acp{ATE}. By letting
$\alpha(X, \vec{W})=X/P(X=1\mid \vec{W})-(1-X)/P(X=0\mid \vec{W})$,
\Cref{theorem:rr-covariates} recovers the classic result that we can estimate
the (unweighted) \ac{ATE} by inverse probability weighting. If $X$ is continuous
with density $f_{X}(x\mid \vec{W})$ conditional on $\vec{W}$, letting
$\alpha(x, \vec{W})=-f_{X}'(x\mid \vec{W})/f_{X}(x\mid \vec{W})$
recovers the average derivative $E[g'(X, \vec{W})]$.

For both of these special cases, there is a wealth of papers studying how to
best implement regression-based or weighting-based approaches to estimating
$\theta(\omega)$, or combinations of both. Recent influential results in the
cross-sectional literature \citep[e.g.,][]{ceinr22,ccddhnr18} highlight the
advantages of combining both approaches using the Neyman orthogonal moment
condition
\begin{equation}\label{eq:doubly_robust_moment}
  \theta(\omega)= E[\mu(X, \vec{W}, g)
  +\alpha(X, \vec{W})(Y-g(X, \vec{W}))],
\end{equation}
where $\mu(X, \vec{W}, g)=g(1,\vec{W})-g(0,\vec{W})$ for the \ac{ATE} and
$\mu(X, \vec{W}, g)=g'(X, \vec{W})$ for the average derivative. This moment
condition is orthogonal in the sense that it is insensitive to small
perturbations in $g$, in contrast to the regression-based moment condition
$\theta(\omega)= E[\mu(X, \vec{W}, g)]$. As a result, an orthogonal
method-of-moments estimator based on~\eqref{eq:doubly_robust_moment} that plugs
in first-stage estimates of $g$ and $\alpha$ can be viewed as a debiased version
of the plug-in estimator utilizing the regression-based moment
condition. Actually, the moment condition~\eqref{eq:doubly_robust_moment} is not
only orthogonal, but also \emph{doubly robust}---insensitive to large
perturbations in either $g$ or $\alpha$ so that the orthogonal method-of-moments
estimator remains consistent so long as any \emph{one} of the first-stage
estimators is consistent for $\alpha$ or $g$, even if the other estimator is
inconsistent. For the binary treatment case, the orthogonal method-of-moments
estimator corresponds to the classic augmented inverse probability weighted
estimator of \citet{rrz94}.

For i.i.d.\ data, it is popular to combine the orthogonal moment condition with
cross-fitting \citep{ceinr22,ccddhnr18}.\footnote{Take a sample sum of the
  moment condition~\eqref{eq:doubly_robust_moment} over the first half of the
  sample, plugging in estimates $\hat{\alpha}_{2}$ and $\hat{\gamma}_{2}$ of
  $\alpha$ and $\gamma$ constructed using the second half of the sample, and add
  to it a sample sum of the moment condition over the second half of the sample,
  but where we plug in estimates $\hat{\alpha}_{1}$ and $\hat{\gamma}_{1}$ from
  the \emph{first} half of the sample.} This allows for regularity conditions
that are weak enough to accommodate a variety of first-step estimators of
$\alpha$ and $\gamma$, including kernels, series, as well as lasso, random
forests, or other machine learning estimators, provided that these estimators
converge sufficiently fast. Because of this flexibility, the approach is known
as debiased machine learning. \Citet{cns22} and \citet{HiWa21}
develop alternatives to this approach that bypass the need to explicitly
estimate the Riesz representer $\alpha$.

These approaches all deliver estimators of $\theta(\omega)$ that converge, under
appropriate regularity conditions, at the usual parametric rate (square root of
sample size) even if the first-stage estimators are based on complicated
nonparametric or machine learning algorithms. Recent work by \citet{goncalves24nonparametric},
\citet{BaWe24}, and \citet{ballarin2025} aims to adapt these
semiparametric approaches to time series contexts. One
potential worry is that even in the absence of covariates, given the small
sample sizes typically available in macroeconomic applications, estimates of
average marginal effects relying on machine-learning or nonparametric first-step
estimates of the shock density and the structural function may yield estimates
that are too noisy and sensitive to the choice of first-stage tuning parameters.
When covariates are needed to argue that the observed variable $X$ is exogenous,
the data requirements become even more severe.

The practical challenges associated with fully nonparametric estimation
motivates studying what the simple \ac{OLS} local projection~\eqref{eqn:lp_reg}
estimates when the true regression function is nonlinear. Extending the analysis
of \Cref{sec:ident-with-contr}, we now allow the researcher to control for
covariates more flexibly by considering the partially linear regression
\begin{equation*}
  Y=X\beta+\gamma(\vec{W})+\text{residual},\quad\text{where}\quad\gamma\in\Gamma,
\end{equation*}
and $\Gamma$ is a linear space of control functions that contains the constant
function $1$. This covers the case with a linear adjustment by letting
$\Gamma=\{a+\vec{w}'\vec{b}\colon
a\in\mathbb{R}, \vec{b}\in\mathbb{R}^{\dim(\vec{w})}\}$ be the class of linear
functions of $\vec{w}$, as well as the semiparametric partially linear model
that lets $\Gamma$ be a large class of ``nonparametric'' functions. By the
projection theorem, the estimand $\beta$ in this regression is given by
\begin{equation}\label{eq:ols-covariates}
  \beta \equiv \frac{E[(X-\pi(\vec{W}))g(X, \vec{W})]}{\var(X-\pi(\vec{W}))},
\end{equation}
where
$\pi(\vec{W}) \equiv
\argmin_{\gamma_{0}\in\Gamma}E[(X-\gamma_{0}(\vec{W}))^{2}]$ denotes the
projection of $X$ onto $\Gamma$ (if $X$ and $\vec{W}$ are independent, the
estimand~\eqref{eq:ols-covariates} reduces to that in~\eqref{eqn:lp_estimand};
we assume the projection exists). We denote the true conditional expectation
(the propensity score, if $X$ is binary) by
$\pi^{*}(\vec{W})\equiv E[X\mid \vec{W}]$.

The next result uses \Cref{theorem:rr-covariates}
to generalize \Cref{thm:obsshock} to the case with covariates.
\begin{prop}\label{theorem:ols_covariates}
  Let $\omega^{*}(x, \vec{W}) \equiv E[\1{X\geq x}(X-\pi^{*}(\vec{W}))\mid \vec{W}]$, and
  suppose $X$ has finite second moments, and that $\var(X-\pi(\vec{W}))>0$. Suppose
  that either (a) $\pi=\pi^{*}$; or else (b) for some
  $\gamma_{0}, \gamma_{1}\in\Gamma$, and some weights $\lambda(x, \vec{w})$
  such that
  $\int \lambda(x, \vec{w})\, dx=\pi^{*}(\vec{w})+\gamma_{1}(\vec{w})$,
  $E[g(X, \vec{W})\mid \vec{W}=\vec{w}]=\gamma_{0}(\vec{w})+ \int \lambda({x},
  \vec{w}) g'(x, \vec{w})\, dx$ for almost all $\vec{w}$.

  Furthermore, assume that conditional on $\vec{W}$, the following holds
  almost surely: (i) the support of $X$ is contained in a (possibly unbounded)
  interval $I_{\vec{W}}\subseteq \mathbb{R}$; and (ii) $g(\cdot, \vec{W})$ is
  locally absolutely continuous on $I_{\vec{W}}$. Finally, assume that (iii)
  $E[\int\abs{\omega^{*}(x, \vec{W})g'(x, \vec{W})}\, dx]<\infty$ and
  $E[\abs{g(X, \vec{W})(X-\pi^{*}(\vec{W}))}]<\infty$.
  Then the estimand~\eqref{eq:ols-covariates} satisfies
   \begin{equation*}
     \beta=\theta(\omega), \quad \text{where} \quad
     \omega(x, \vec{W}) \equiv
     \frac{\omega^{*}(x, \vec{W})+ (\pi^{*}(\vec{W})-\pi(\vec{W})) \lambda({x}, \vec{W})}{
       \var(X-\pi(\vec{W}))},
  \end{equation*}
  where, if condition (a) holds, we let $\lambda(x, \vec{w})=0$.

  The weights integrate to one: $E[\int \omega(x, \vec{W})\, dx]=1$. A sufficient
  condition for the weights to be non-negative is that condition (a) holds, in
  which case $\omega^{*}(x, \vec{w})$ is hump-shaped as a function of $x$ for
  almost all $\vec{w}$: monotonically increasing from $0$ to its maximum for
  $x=\pi^{*}(\vec{w})$, and then monotonically decreasing back to $0$.
\end{prop}

To interpret this result, it is useful to first consider the case where the class $\Gamma$ of control
functions is rich enough so that condition (a) holds. This is trivially the case,
for instance, if $\vec{W}$ just consists of a single set of fixed effects. In
this case, we obtain an analogue of \Cref{thm:obsshock}: (partially) linear
regression identifies a weighted average of marginal effects, with hump-shaped
weights---this holds regardless of the distribution of $X$, be it discrete,
continuous or mixed. If the conditional distribution of $X$ given $\vec{W}$
is continuous, then the weighting
$\omega(x, \vec{W})=\omega^{*}(x, \vec{W})/\var(X-\pi^{*}(\vec{W}))$
varies smoothly with $x$; if there are mass points, as in the case of a discrete
or mixed distribution, then the weight function jumps discontinuously at the mass
points. If the treatment $X$ is discrete with support $0, 1, \dotsc$, we recover
the result in \citet{Angrist1999} that regression estimates a weighted average
of the causal effects of increasing $X$ by one unit,
\begin{equation*}
  \beta=\frac{E\left[\sum_{s=1}^{\infty} E[\1{X\geq s}(X-\pi^{*}(\vec{W}))\mid \vec{W}]
      (g(s, \vec{W})-g(s-1,\vec{W}))\right]}{
    E\left[\sum_{s=1}^{\infty} E[\1{X\geq s}(X-\pi^{*}(\vec{W}))\mid \vec{W}]\right]
  }.
\end{equation*}

Now suppose that condition (a) is violated, because the specification for
$\Gamma$ is not sufficiently flexible to model the true propensity score
$\pi^{*}(\vec{W})$. In general, this may lead to omitted variable bias, as
the partially linear model may not be sufficiently flexible to account for all
confounding due to $\vec{W}$. Condition (b) prevents this scenario, ensuring
that any bias due to confounding is accounted for. As a simple example of when
the condition holds, consider the case where the true conditional mean
function has a multiplicative form:
$g(X, \vec{W})=X g'(\vec{W})+\gamma_{0}(\vec{W})$, with
$\gamma_{0}\in\Gamma$. The partially linear model is misspecified, because the
marginal effect is not constant, but varies with $\vec{W}$. But condition (b)
holds with $\gamma_{1}(\vec{w})=0$ and
$\lambda(x, \vec{w})=\pi^{*}(\vec{w})\varphi(x)$, where $\varphi(x)$ is an
arbitrary density function. Because the marginal effect varies only with
$\vec{W}$ but not with $X$, \Cref{theorem:ols_covariates} simplifies to
\begin{equation*}
  \beta= \frac{E\left[(E[X^{2}\mid \vec{W}]-\pi(\vec{W})\pi^{*}(\vec{W}))g'(\vec{W})\right]}{\var(X-\pi^{*}(\vec{W}))}.
\end{equation*}
If $\pi=\pi^{*}$, we obtain a generalization of the \citet{angrist98} result for binary treatments: the
weights are proportional to the conditional variance of $X$,
$\var(X\mid \vec{W})$. But if $\pi\neq \pi^{*}$, the weight function may be
negative for some values of $\vec{W}$.
Consider, for instance, a panel data scenario where $\Gamma$ is linear, and
$\vec{W}$ consists of unit and time fixed effects. Then the assumption
$\gamma_{0}\in\Gamma$ amounts to a parallel trends assumption: in the absence of
the treatment, the average differences in outcomes for different units are
constant and do not depend on the time period. If the treatment effects are
heterogeneous, so that $g'(\vec{W})$ depends on the unit and time period,
then this two-way fixed effects regression still estimates a weighted average of
marginal effects, but with weights that are negative if
$E[X^{2}\mid \vec{W}]<\pi(\vec{W})\pi^{*}(\vec{W})$. In the context of
a binary treatment and two-way fixed effects regressions, this result has been
noted in \citet{dCDH20} and \citet{gb21}. \Cref{theorem:ols_covariates}
generalizes this to an arbitrary treatment distribution and a general regression
specification.

Another example where condition (b) of \Cref{theorem:ols_covariates} holds is when $X$ is bounded below by some baseline value (say, 0), $X\geq 0$, and the baseline outcome model is correctly specified, $g(0,\vec{W})\in\Gamma$. Then condition (b) holds
  with $\gamma_{0}(\vec{W})=g(0,\vec{W})$, $\gamma_{1}=0$, and
  $\lambda(x, \vec{W})=P(X\geq x\mid \vec{W})$. In this case, the weights
  simplify to
  $\omega(x, \vec{W})=E[\1{X\geq x}(X-\pi(\vec{W}))\mid \vec{W}]/\var(X-\pi(\vec{W}))$.

  Thus, if the marginal effects are constant, the partially linear model is
  doubly robust: the regression estimand is consistent for this constant
  treatment effect so long as either $\pi^{*}\in\Gamma$ or
  $g(0,\vec{W})\in\Gamma$, as noted, for instance, in \citet{RoMaNe92}. But this
  double robustness doesn't fully extend to the case with heterogeneous marginal
  effects. If the researcher gets the model for $X$ right, in the sense that
  $\pi=\pi^{*}$, then the partially linear regression estimates an average
  marginal effect. However, if the researcher gets it wrong, and only gets the
  outcome model under no treatment right, so that only condition (b) of
  \Cref{theorem:ols_covariates} holds, then weights on some of the true marginal
  effects $g'(X, \vec{W})$ may be negative, risking a sign-reversal. This
  asymmetry has been noted by \citet{GoldsmithPinkham2024} for the case with a
  binary treatment $X$.\footnote{\citet{GoldsmithPinkham2024} also show that in
    regressions on multiple mutually exclusive treatment indicators, the
    regression estimand on a given treatment contains an additional
    contamination bias term corresponding to non-convex average effects of the
    other treatments.} \Cref{theorem:ols_covariates} shows that the result is
  general. The upshot of this asymmetry is that in cases where the treatment
  variable $X$ is only conditionally exogenous---whether the data is
  cross-sectional, panel, or time series---it pays off to conduct a sensitivity
  analysis with respect to the functional form of the control specification, in
  addition to the standard sensitivity analysis with respect to the set of
  controls.

\section{Conclusion}\label{sec:concl}

We have shown that conventional linear methods for identifying causal effects in applied time series analysis based on observed shocks or proxies are robust to misspecification: they estimate a positively weighted average of the true nonlinear causal effects, irrespective of the extent of nonlinearities in the underlying \ac{DGP}. By contrast, identification approaches that exploit heteroskedasticity or non-Gaussianity of latent shocks are highly sensitive to violations of the assumed linear functional form of the structural model. Moreover, while linear identification via heteroskedasticity provides some testable restrictions, identification via non-Gaussianity is generally unfalsifiable despite the potential for severe biases.

Our results suggest that it is worthwhile for applied researchers who are interested in average marginal effects to expend the effort involved in constructing direct measures of shocks, or at least proxies that are credibly (approximately) monotonically related to the latent shock of interest. Shock measurement via narrative approaches or detailed institutional knowledge is admittedly highly work-intensive and in many applications may be practically impossible. Nevertheless, when the strategy is feasible, it affords an insurance against functional form misspecification which is not matched by the other identification approaches that we analyze. This is similar to the robustness of linear regression for estimating average treatment effects in randomized control trials relative to estimators that rely on parametric adjustment for nonlinearities or selection in observational data. Our results do not directly speak to other identification approaches like identification via long-run or sign restrictions, and we leave these for future work.\footnote{As a referee pointed out to us, it may be possible to analyze long-run restrictions by extending \cref{thm:proxy} to apply to the instrumental variable specification in \citet{Shapiro1988}. However, since the derivations in the latter paper rely on linearity of the structural function, it is not immediately clear what assumptions would guarantee convex weights.}

When reporting impulse responses from linear specifications with observed
shocks, we recommend that researchers routinely report the implicit weight
function, which is easy to compute via standard regression software. If the
weight function associated with conventional local projections or \acp{VAR} is
deemed to be unattractive, one can target other causal summaries as discussed in
\Cref{sec:mte}, though further analysis is required on the best practices for
doing so in macroeconomic applications.

More broadly, we hope that our paper will boost the agenda spearheaded by \citet{White2006}, \citet{White2009},
\citet{Angrist2011}, \citet{Angrist2018}, \citet{Rambachan2021}, and \citet{goncalves24nonparametric,Goncalves2024} that seeks
to draw lessons for macroeconometrics from the microeconometric treatment effect
literature. The nonparametric framework used in the treatment effect literature
contains useful lessons for empirical work in macroeconomics, even though conventional approaches to nonparametric estimation or debiased machine learning methods are impractical due to the much smaller data sets typical in macroeconomics.

\clearpage%
\phantomsection%
\addcontentsline{toc}{section}{References}
\bibliography{ref}

\end{document}


\title{\texorpdfstring{\vspace{-1.5\baselineskip}}{} Online Supplement for ``Dynamic Causal Effects in a
  Nonlinear World: the Good, the Bad, and the Ugly''} \author{Michal Koles\'{a}r \\
  Princeton University \and Mikkel Plagborg-M{\o}ller \\ Princeton University}
\date{\texorpdfstring{\bigskip}{ }\today}
\maketitle

\begin{appendices}
\crefalias{section}{sappsec}
\crefalias{subsection}{sappsubsec}
\crefalias{subsubsection}{sappsubsubsec}

\section{Further results}\label{app:results}

\subsection{Integrals of the weight function}\label{app:weight_integr}

The following lemma provides an identification result for integrals of the
causal weight function $\omega_{X}$ defined in \Cref*{sec:obsshock}.

\begin{lem}\label{thm:weight_integral}
Let $\omega_X$ be given by (\ref*{eqn:obsshock_weights}). Assume that $E[X_t^2]<\infty$. Let $\underline{x}, \overline{x}$ be constants such that $-\infty \leq \underline{x} < \overline{x} \leq \infty$. Then
\begin{equation*}
  \int_{\underline{x}}^{\overline{x}} \omega_X(x)\, dx =
  \frac{\cov\left(\max\lbrace \min\lbrace X_t, \overline{x} \rbrace
  , \underline{x} \rbrace, X_t\right)}{\var(X_t)}.
\end{equation*}
\end{lem}
\begin{proof}
By Fubini's theorem and linearity of the covariance operator,
\begin{equation*}
  \int_{\underline{x}}^{\overline{x}}
  \cov(\1{X_t \geq x}, X_t)\, dx = \cov\left(\int_{\underline{x}}^{\overline{x}}\1{X_t \geq x}\, dx, X_t\right).
\end{equation*}
Considering separately the three cases $X_t < \underline{x}$,
$X_t \in [\underline{x}, \overline{x}]$, and $X_t>\overline{x}$, it can be
verified that
\begin{equation*}
  \int_{\underline{x}}^{\overline{x}}\1{X_t \geq x}\, dx = \max\lbrace
  \min\lbrace X_t, \overline{x} \rbrace, \underline{x} \rbrace - \underline{x}. \qedhere
\end{equation*}
\end{proof}
\noindent Note that the lemma holds even if $X_t$ has a discrete distribution
(e.g., the empirical distribution). It implies in particular that the
OLS-estimated weight function discussed in \Cref*{sec:obsshock} integrates to 1
across all $x \in \mathbb{R}$ in finite samples.

\subsection{Identification with instruments under
  endogeneity}\label{app:ident-with-instr}

We now generalize the setup in \Cref*{sec:proxy} by allowing $X_{t}$ to be
endogenous and incorporating covariates $\vec{W}_{t}$. In particular, we
retain the nonparametric structural model~(\ref*{eqn:causal}), but drop the
independence assumption~(\ref*{eqn:indep}) and the selection-on-observables
assumption~(\ref*{eqn:selection_obs}). Let
\begin{equation}\label{eq:first_stage}
  X_{t}=\xi(Z_{t}, \vec{W}_{t}, \tilde{V}_{t})
\end{equation}
denote the first-stage equation, with $\tilde{V}_{t}$ corresponding to the
unobservable determinants of $X_{t}$. To accommodate a variety of alternatives
to the classic \citet{ImAn94} monotonicity assumption, we follow
\citet{strlb17} and suppose there is a vector $\vec{V}_{t}$ (not necessarily
observable) containing the covariates $\vec{W}_{t}$ that is a sufficient
statistic for endogeneity in the sense that
\begin{equation}\label{eq:indep}
  E[\psi_{h}(x, \vec{U}_{h, t+h})\mid X_{t}, Z_{t}, \vec{V}_{t}]
  =\Psi_{h}(x, \vec{V}_{t}),
\end{equation}
where
\begin{equation*}
  \Psi_{h}(x, \vec{v}) \equiv E[\psi_{h}(x, \vec{U}_{h, t+h})\mid \vec{V}_{t}=\vec{v}]
\end{equation*}
denotes the \emph{marginal treatment response function}. We also make the
exclusion restriction that the expectation of $Z_{t}$ conditional on
$\vec{V}_{t}$ depends only on $\vec{W}_{t}$. To ensure it is sufficient to
control for the covariates linearly, we further assume the conditional
expectation is linear:
\begin{equation}\label{eq:linear_control}
  E[Z_{t}\mid \vec{V}_{t}]=\vec{W}_{t}'\vec{\gamma}.
\end{equation}
We assume implicitly that the conditional expectations in
\cref{eq:indep,eq:linear_control} are well-defined.

This setup accommodates several scenarios. On the one hand, including more
variables in $\vec{V}_{t}$ makes~\cref{eq:indep} less restrictive; on the other
hand, as we will see shortly, it requires stronger conditions to ensure
non-negative weights in the instrumental variables estimand. As a leading case,
we may put $\vec{V}_{t}=(\vec{W}_{t}, \tilde{V}_{t})$. Then
\cref{eq:indep,eq:linear_control} hold if the instrument is conditionally
randomly assigned (and $E[Z_t \mid \vec{W}_t]$ is linear). In this case, with
a binary shock $X_{t}$, the difference $\Psi_h(1,\vec{v})-\Psi_h(0,\vec{v})$ corresponds to
the marginal treatment effect of \citet{HeVy99,HeVy05}. Second, under selection
on observables, we can simply put $\vec{V}_{t}=\vec{W}_{t}$. In this case,
\cref{eq:indep} states that $Z_{t}$ is a valid proxy for $X_{t}$, analogously to \cref*{eqn:proxy}; moreover, $\Psi_{h}$ equals the conditional average structural function. Third, we may put
$\mathbf{V}_{t}=(\vec{W}_{t}, \vec{U}_{h, t+h})$, in which case~\cref{eq:indep} holds
trivially, and $\Psi_{h}=\psi_{h}$. See \citet[Section 6]{strlb17} for examples of other choices for
$\vec{V}_{t}$ (which is denoted by $\vec{U}$ in their notation) when $X_{t}$ is
assumed to be binary.

Under equations~(\ref*{eqn:causal}), \eqref{eq:indep}, and
\eqref{eq:linear_control}, it follows by the Frisch-Waugh theorem and iterated
expectations that the coefficient on $Z_{t}$ in a linear ``reduced-form''
regression of $Y_{t+h}$ onto $Z_{t}$ and a vector of controls $\vec{W}_{t}$ is
given by
\begin{equation} \label{eqn:iv_estimand}
  \tilde{\beta}_{h}=\frac{E[(Z_{t}-\vec{W}_{t}'\vec{\gamma})\Psi_{h}(X_{t},
    \vec{V}_{t})]}{E[(Z_{t}-\vec{W}_{t}'\vec{\gamma})^{2}]}.
\end{equation}
As in \Cref*{sec:ident-with-contr-1}, we assume that the support of $X_{t}$
conditional on $\vec{V}_{t}$ is contained in an interval $I_{\vec{V}_{t}}$. If
there are gaps in the support of $X_{t}$, such as when $X_{t}$ is discrete, we
assume that we can extend $\Psi_{h}(\cdot, \vec{V}_{t})$ to $I_{\vec{V}_{t}}$
such that the extension is locally absolutely continuous.
Applying~\Cref*{theorem:rr-covariates} with $\vec{V}_{t}$ playing the role of
the covariates $\vec{W}$ then yields the following result:

\begin{prop}\label{theorem:iv_endogeneity}
  Suppose \cref{eq:linear_control} holds and that
  $E[(Z_{t}-\vec{W}_{t}'\vec{\gamma})^{2}]$ is positive and finite. Let
  $\alpha(\vec{V}_{t}, X_{t})=E[Z_{t}-\vec{W}_{t}'\vec{\gamma}\mid
  \vec{V}_{t}, X_{t}]$. Suppose also that conditional on $\vec{V}_{t}$, the
  following holds almost surely: (i) the support of $X_{t}$ is contained in a
  (possibly unbounded) interval $I_{\vec{V}_{t}}\subseteq \mathbb{R}$; and (ii)
  $\Psi_{h}(\cdot, \vec{V}_{t})$ is locally absolutely continuous on
  $I_{\vec{V}_{t}}$. Suppose also that (iii) there exists a function
  $x_{0}(\vec{V}_{t})\in I_{\vec{V}_{t}}$ such that
  $E[\abs{\alpha(\vec{V}_{t}, X_{t})\int_{x_{0(\vec{V}_{t})}}^{X_{t}}
    \abs{\Psi_{h}'(x, \vec{V}_{t})}\, dx}]<\infty$; and that (iv)
  $E[\abs{\alpha(\vec{V}_{t}, X_{t})}(1+\abs{\Psi_{h}(X_{t},
    \vec{V}_{t})})]<\infty$. Then the estimand~\eqref{eqn:iv_estimand} satisfies
   \begin{equation*}
     \tilde{\beta}_{h}= E\left[\int\omega(x, \vec{V}_{t})\Psi_{h}'(x, \vec{V}_{t})\, dx\right],
  \end{equation*}
  where
  $\omega(x, \vec{v}) \equiv E[\1{X_{t}\geq x}
  (Z_{t}-\vec{W}_{t}'\vec{\gamma})\mid \vec{V}_{t}=\vec{v}]/\var(Z_{t}-\vec{W}_{t}'\vec{\gamma})$, and $\Psi_h'(x, \vec{v})$ denotes the partial derivative with respect to $x$.
\end{prop}

\Cref{theorem:iv_endogeneity} shows that the reduced-form regression of $Y_{t}$
onto $Z_{t}$ identifies a weighted average of derivatives of the marginal
treatment response function.\footnote{An analogous application of
  \Cref*{theorem:rr-covariates} to a linear ``first-stage regression'' of
  $X_{t}$ onto $Z_{t}$ and a vector of controls $\vec{W}_{t}$ shows that it
  identifies the integral of the weights, $E[\int\omega(x, \vec{V}_{t})dx]$, so
  that the weights in the associated instrumental variables regression integrate
  to one.} A sufficient condition ensuring non-negative weights
$\omega(x, \vec{v})$ is the stochastic monotonicity condition that
$E[Z_{t}\mid X_{t}=x, \vec{V}_{t}]$ is almost surely monotone increasing (or
decreasing) in $x$.

Applying the result with $\vec{V}_{t}=(\vec{W}_{t}, \tilde{V}_{t})$ generalizes
Theorem 1 of \citet{AnGrIm00} in several ways: we don't require
differentiability of the potential outcome function $\psi_{h}$, only of the
marginal treatment response function; $X_{t}$ is not required to be
continuous---it may be discrete or mixed; $Z_{t}$ is not restricted to be
binary; we impose no structure on the first-stage equation; and finally, we
impose only very weak moment conditions. In this case, the stochastic
monotonicity assumption is equivalent to the first-stage monotonicity condition
that $\xi(z, \vec{W}_{t}, \tilde{V}_{t})$ is increasing in $z$: this corresponds
to Assumption 4 in \citet{AnGrIm00} if $z$ is binary.
Under this condition, $\tilde{\beta}_{h}$ can be interpreted as identifying a
weighted average of marginal effects for compliers. Let
$\mathcal{C}_{\vec{w}}$ collect all $\tilde{v}$ in the support of
$\tilde{V}_{t}$ such that $\xi(\cdot, \vec{w}, \tilde{v})$ is not constant.
Following \citet{AnImRu96}, we refer to the set $\mathcal{C}_{\vec{w}}$ as the set of
\emph{compliers}, since if $\tilde{V}_{t}\in \mathcal{C}_{\vec{W}_{t}}$, the shock $X_{t}$
complies with the instrument assignment in the sense that it increases with
$Z_{t}$. If $\tilde{V}_{t}\not\in\mathcal{C}_{\vec{W}_{t}}$, variation in the instrument $Z_{t}$ has
no impact on $X_{t}$, and hence $\omega(x, \vec{W}_{t}, \tilde{V}_{t})=0$ for all
$x$. Thus, the estimand $\tilde{\beta}_{h}$ only places positive weight on the
marginal effect $\Psi_{h}'(x, \vec{W}_{t}, \tilde{V}_{t})$ for compliers.

As discussed in \citet{strlb17} in the context with a binary treatment, the
first-stage monotonicity assumption may be too strong in some contexts. In such
scenarios, other choices of $\vec{V}_{t}$ may be preferable, such as setting
$\vec{V}_{t}=(\vec{W}_{t}, \vec{U}_{h,t+h})$. For this choice of $\vec{V}_{t}$,
\Cref{theorem:iv_endogeneity} generalizes Proposition 1 in \citet{bh24} by
allowing $X_{t}$ to have full support, and dropping the requirement that $X_{t}$
be continuous.

If $X_{t}$ is exogenous, we may set $\vec{V}_{t}=\vec{W}_{t}$, which both
weakens the condition ensuring non-negative weights and broadens the
interpretation of the estimand. In particular, now
$\Psi_{h}'(x, \vec{v})=\partial E[\psi_h(x, \vec{U}_{h, t+h})\mid
\vec{W}_{t}=\vec{w}]/\partial x$ gives the overall marginal effect, not just the
effect for compliers. Also, now the stochastic monotonicity condition requires
only that $E[Z_{t}\mid X_{t}=x, \vec{W}_{t}]$ is increasing in $x$. Without
covariates, this reduces to the condition that $\zeta(x)=E[Z_t \mid X_t=x]$ is
monotone, as discussed in \Cref*{sec:proxy}. This is clearly weaker than the
\citet{AnGrIm00} condition that $E[Z_{t}\mid X_{t}=x, \vec{W}_{t}, \tilde{V}_{t}]$
is increasing in $x$: we only require this to be true \emph{on average} over
$\tilde{V}_{t}$ rather than for almost all realizations of $\tilde{V}_{t}$. This
condition holds for many measurement error models for $Z_{t}$, even though the
stronger first-stage monotonicity condition may be violated.

\subsection{Identification via heteroskedasticity: linear case}\label{app:hetero_identif_linear}

Here we derive the linear identification result (\ref*{eqn:hetero_identif_linear}), following \citet{Rigobon2004} and \citet{Lewbel2012}. Note first that
\begin{equation*}
  \begin{split}
    E[Z \mid \vec{U}]
    &=
      E[(\theta_1 X + \gamma_1(\vec{U}))(D-E[D]) \mid \vec{U}] \\
    &= \theta_1 E[X(D-E[D]) \mid \vec{U}] + \gamma_1(\vec{U})E[D-E(D) \mid \vec{U}] \\
    &= \theta_1 \cov(X, D) + \gamma_1(\vec{U})E[D-E(D)] \\
    &= 0.
  \end{split}
\end{equation*}
Hence,
\begin{equation*}
  \cov(\vec{Y}, Z) = \vec{\theta}\cov(X, Z) + \cov(\vec{\gamma}(\vec{U}), Z) = \vec{\theta}\cov(X, Z),
\end{equation*}
and the claim (\ref*{eqn:hetero_identif_linear}) follows, provided that $\cov(X, Z) \neq 0$. The latter holds if $\theta_1 \neq 0$ and $\cov(X^2,D) \neq 0$, since
\begin{align*}
  \pushQED{\qed}
  \cov(X, Z) &= E[X(\theta_1 X + \gamma_1(\vec{U}))(D-E[D])] \\
            &= \theta_1 \cov(X^2,D) + \cov(X, D)E[\gamma_1(\vec{U})] \\
            &= \theta_1 \cov(X^2,D). \qedhere
              \popQED
\end{align*}

\subsection{Details for \texorpdfstring{\Cref*{exm:nongauss_counter2}}{Example~\ref*{exm:nongauss_counter2}}}\label{app:nongauss_counter2}
Let $\tilde{U}_1$ and $\tilde{U}_2$ be independent uniforms on $[0,1]$. By the Box-Muller transform, the two variables
\begin{equation*}
  \tilde{Y}_1 \equiv \sqrt{-2\log \tilde{U}_1}\cos(2\pi \tilde{U}_2), \quad \tilde{Y}_2 \equiv \sqrt{-2\log \tilde{U}_1}\sin(2\pi \tilde{U}_2),
\end{equation*}
have a bivariate standard normal distribution.

Define $X \equiv \log(-2\log \tilde{U}_1)$ and
$U \equiv \log\cos^2(2\pi\tilde{U}_2)$, so that $X$ and $U$ are independent and
non-Gaussian. By construction, the following two variables are independent:
\begin{equation*}
  Y_1 \equiv \log \tilde{Y}_1^2 = X + U, \quad Y_2 \equiv \log\tilde{Y}_2^2 = X + \gamma(U),
\end{equation*}
where
\begin{equation*}
  \gamma(u) \equiv \log\left(1-\exp(u)\right), \quad u<0,
\end{equation*}
and we have used that
$\exp(U) = \cos^2(2\pi\tilde{U}_2) = 1-\sin^2(2\pi\tilde{U}_2)$. Note that in
this example, the shocks $X$ and $U$ do not have mean zero as commonly assumed
in the literature, but this is easily rectified by just subtracting off their
means in the calculations.\qed%

\subsection{Additional empirical estimates of causal weights}\label{app:weight_empir}

Complementing the results for government spending shocks in \Cref*{fig:gov} (\Cref*{sec:obsshock}), \Cref{fig:tax,fig:tech,fig:mon} show estimated causal weight functions for several identified tax shocks, technology shocks, and monetary policy shocks. The data is obtained from the replication files for \citet{Ramey2016}, as discussed in \Cref*{sec:obsshock}. While many of the shocks yield approximately symmetric weight functions, the \citet{Romer2010} and \citet{Mertens2014} tax shocks are both skewed towards tax cuts, while the \citet{Christiano1999} and \citet{Gertler2015} monetary shocks are skewed towards interest rate cuts. As discussed in \Cref*{sec:obsshock}, this is important to keep in mind when using impulse response estimates to discipline structural models that feature asymmetries.

\begin{figure}[p]
\centering
\textsc{Empirical weight functions: tax shocks} \\
\includegraphics[width=\linewidth,clip=true,trim=0 1em 0 0]{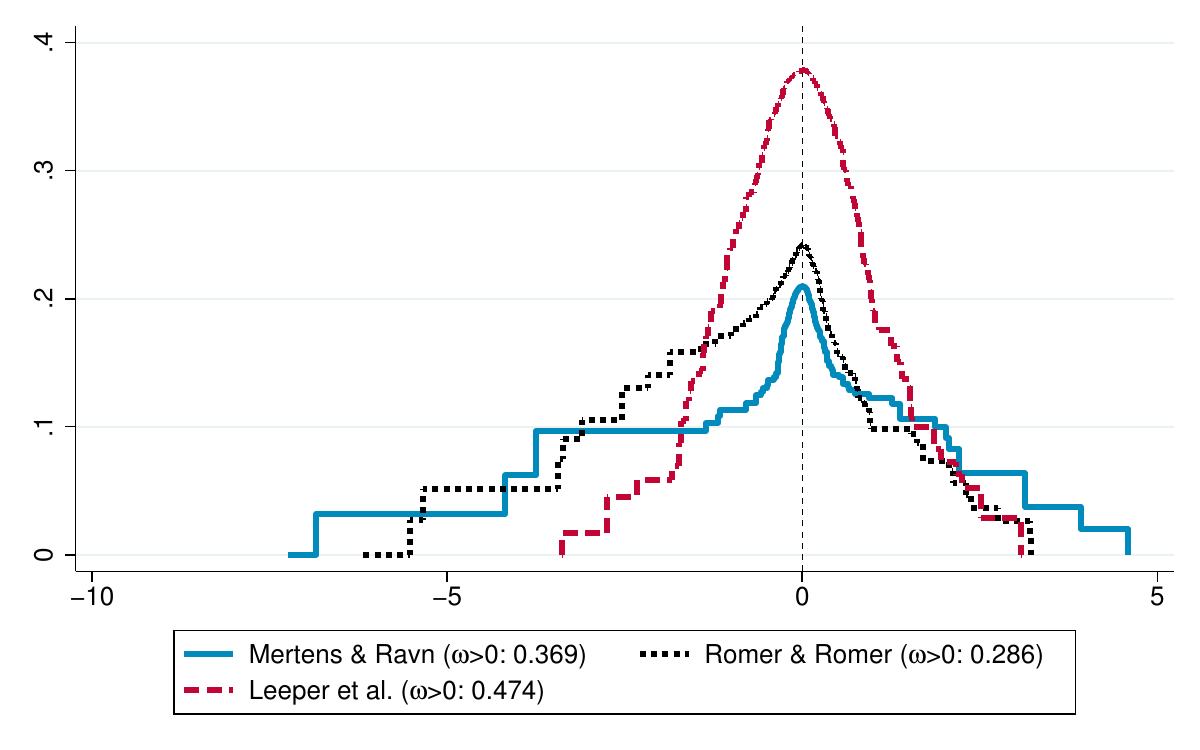}
\caption{Estimated causal weight functions $\omega_X$ for tax shocks obtained from the replication files for \citet{Ramey2016}, quarterly data. Horizontal axis in units of standard deviations. ``$\omega>0$'': total weight $\int_0^\infty \omega_X(x)\, dx$ on positive shocks. Papers referenced: \citet{Mertens2014}, \citet{Romer2010}, \citet{Leeper2012}.}\label{fig:tax}
\end{figure}

\begin{figure}[p]
\centering
\textsc{Empirical weight functions: technology shocks} \\
\includegraphics[width=\linewidth,clip=true,trim=0 1em 0 0]{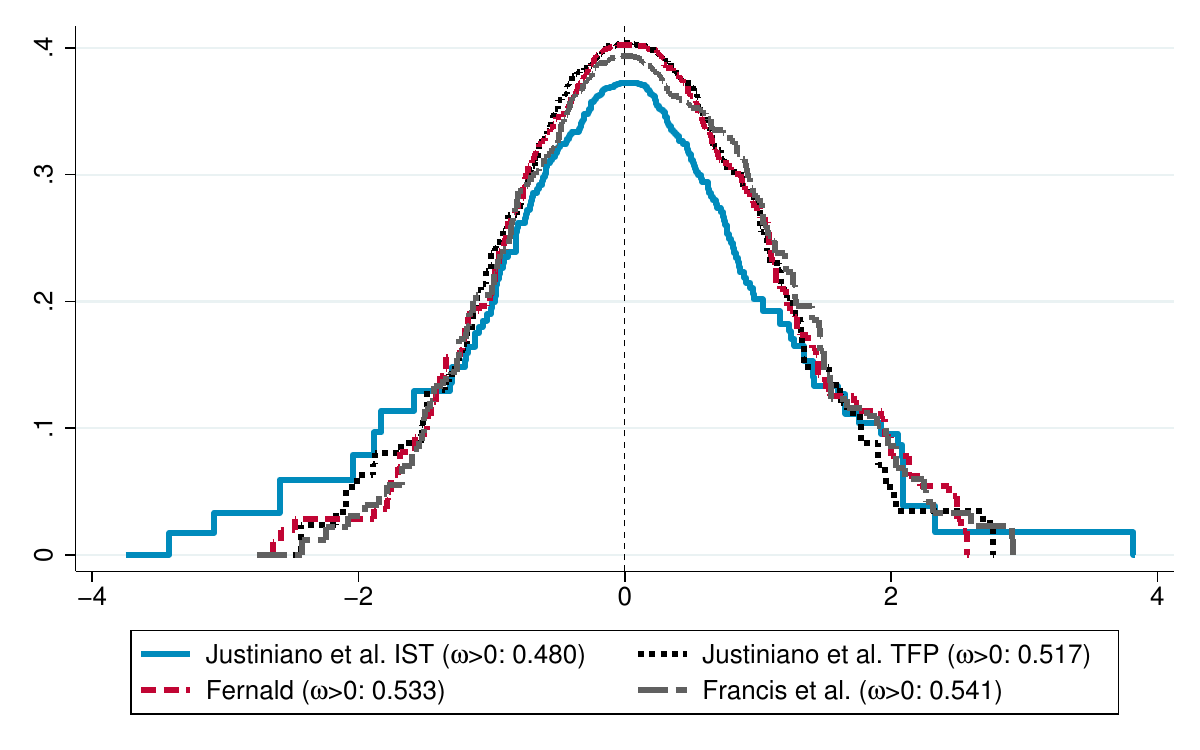}
\caption{Estimated causal weight functions $\omega_X$ for technology shocks obtained from the replication files for \citet{Ramey2016}, quarterly data. Horizontal axis in units of standard deviations. ``TFP'' = total factor productivity. ``IST'' = investment-specific technology. ``$\omega>0$'': total weight $\int_0^\infty \omega_X(x)\, dx$ on positive shocks. Papers referenced: \citet{Justiniano2011}, \citet{Fernald2014}, \citet{Francis2014}.}\label{fig:tech}
\end{figure}

\begin{figure}[p]
\centering
\textsc{Empirical weight functions: monetary policy shocks} \\
\includegraphics[width=\linewidth,clip=true,trim=0 1em 0 0]{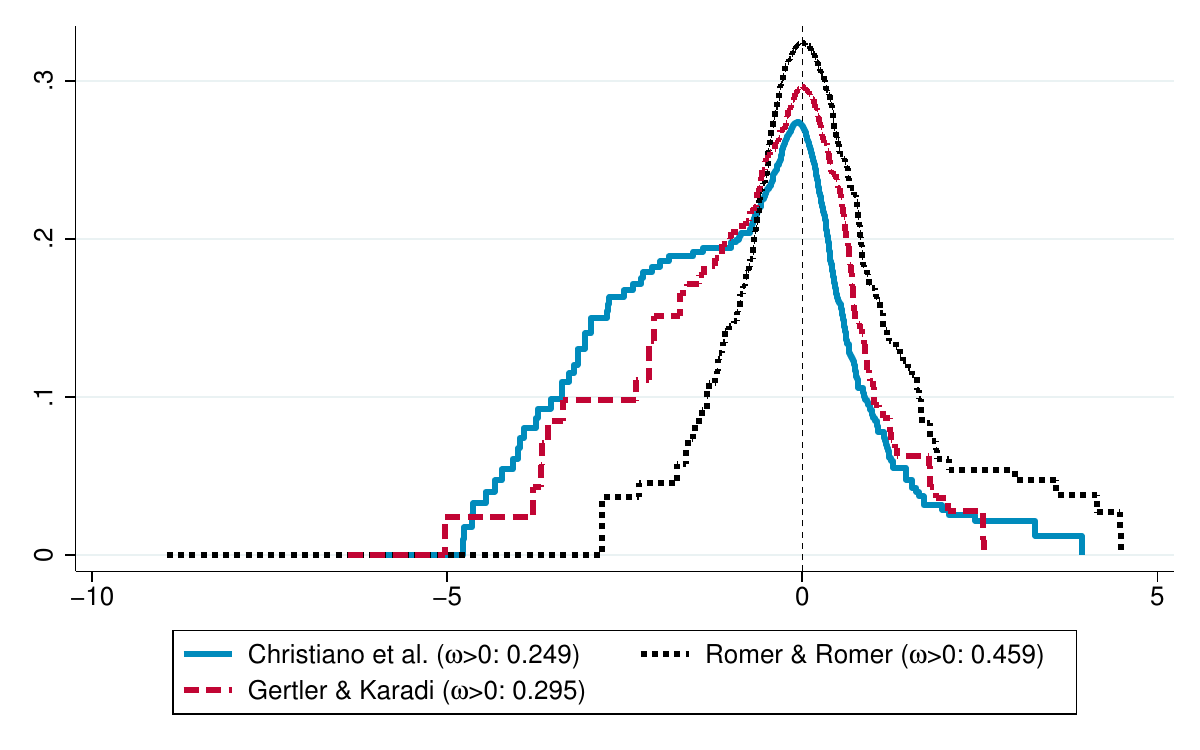}
\caption{Estimated causal weight functions $\omega_X$ for monetary policy shocks obtained from the replication files for \citet{Ramey2016}, quarterly data. Horizontal axis in units of standard deviations. ``$\omega>0$'': total weight $\int_0^\infty \omega_X(x)\, dx$ on positive shocks. Papers referenced: \citet{Christiano1999}, \citet{Romer2010}, \citet{Gertler2015}.}\label{fig:mon}
\end{figure}

\clearpage

\section{Proofs}\label{app:proofs}

\subsection{Auxiliary lemma}

\begin{lem}\label{lemma:finite_integral_omega}
  Suppose that conditions (i)--(iii) of
    \Cref*{theorem:rr} hold. Suppose additionally that for some $\underline{x}, \overline{x}\in I$, $\underline{x}\leq \overline{x}$,
  it holds that either (a) $\alpha(x)$ only changes sign for $x \in [\underline{x}, \overline{x}]$ and $\int_I |\omega(x)g'(x)|\, dx<\infty$, or (b) $g(x)$ is monotone for $x\leq \underline{x}$ and for $x\geq \overline{x}$.  Then condition (iv) of \Cref*{theorem:rr} holds for
  any $x_{0}\in[\underline{x}, \overline{x}]$.
\end{lem}
\begin{proof}
  Bound
  \begin{align*}
  E\left[\left|\alpha(X)\int_{x_0}^X|g'(x)|\, dx\right|\right] &\leq E[\abs{\alpha(X)}]\int_{\underline{x}}^{\overline{x}}|g'(x)|\, dx + E\left[\1{X \geq \overline{x}}\abs{\alpha(X)}\int_{\overline{x}}^X|g'(x)|\, dx\right] \\
  &\quad + E\left[\1{X \leq \underline{x}}\abs{\alpha(X)}\int_X^{\underline{x}}|g'(x)|\, dx\right].
  \end{align*}
  The first term on the right-hand side is finite since $g$ is absolutely
  continuous on $[\underline{x}, \overline{x}]$. Now consider the second term on
  the right-hand side; the third term can be handled analogously. Under
  condition (a), $\alpha(x)$ has the same sign for all $x \geq \overline{x}$, so
  the second term equals
  \[\int_I \1{x \geq \overline{x}} \left|E[\1{X \geq x}\alpha(X)]\right||g'(x)|\, dx \leq \int_I |\omega(x)| |g'(x)|\, dx < \infty.\]
  Under condition (b), since $g(x)$ is monotone for $x \geq \overline{x}$,
  $\int_{\overline{x}}^X\abs{g'(x)}\, dx=\abs{\int_{\overline{x}}^{X}g'(x)\, dx}$, so that the second term on the right-hand side in the first display equals
  \begin{align*}
  E\left[\1{X \geq \overline{x}}\abs{\alpha(X)}\left|\int_{\overline{x}}^X g'(x)\, dx\right|\right] &= E\left[\1{X \geq \overline{x}}\abs{\alpha(X)}\abs{g(X)-g(\overline{x})}\right] \\
  &\leq E[\abs{\alpha(X)g(X)}]+\abs{g(\overline{x})}E[\abs{\alpha(X)}] < \infty. \qedhere
  \end{align*}
\end{proof}

\subsection{Proof of \texorpdfstring{\Cref*{thm:obsshock}}{Proposition~\ref*{thm:obsshock}}}

This is a special case of \cref*{thm:proxy} with $Z=\zeta(X)=X$.
\Cref{thm:weight_integral} implies that the weights integrate to 1.
\qed%

\subsection{Proof of \texorpdfstring{\Cref*{thm:quad_reg}}{Proposition~\ref*{thm:quad_reg}}}

Since $g_{h}'(x)$ is locally absolutely continuous and
$E[\abs{g_{h}''(X_{t})}]<\infty$, by Stein's lemma \citep[Lemma 1]{stein81},
\begin{equation*}
  E[g_{h}''(X_{t})]=E[X_{t}g_{h}'(X_{t})].
\end{equation*}
Since $E[\abs{g_{h}(X_{t})}]<\infty$, another application of Stein's lemma
yields $E[X_{t}g_{h}'(X_{t})+g_{h}(X_{t})]=E[X_{t}^{2}g_{h}(X_{t})]$. Hence,
$\cov(g_{h}(X_{t}), X_{t}^{2})=E[X_{t}g_{h}'(X_{t})]=E[g_{h}''(X_{t})]$. A third
application of Stein's lemma yields $E[g_{h}'(X_{t})]=\cov(X_{t}, g_{h}(X_{t}))$.
The result then follows from the
definitions~(\ref*{eqn:quad_reg_deriv})--(\ref*{eqn:quad_reg_coef}).
\qed%

\subsection{Proof of \texorpdfstring{\Cref*{thm:proxy}}{Proposition~\ref*{thm:proxy}}}

The representation of the estimand follows directly from
\Cref*{theorem:rr} and \Cref{lemma:finite_integral_omega} with
$\alpha(X_{t}) = \zeta(X_{t})-E[Z_{t}]$. Claim (i) for the weights follows from
a simple calculation. Claim (ii) follows from
$\cov(\1{X_t \geq x}, \zeta(X_{t})) = \var(\1{X_t \geq x})\lbrace E[\zeta(X_{t})
\mid X_t \geq x] - E[\zeta(X_{t}) \mid X_t < x] \rbrace$. For the last statement
of the proposition, observe that for $x_{U}>x_{L}$,
$\tilde{\omega}_{Z}(x_{L})-\tilde{\omega}_{Z}(x_{U})$ is proportional to
$E[\1{x_{L}<X_{t}< x_{U}}(\zeta(X_{t})-E[Z_{t}])]$, which is positive if
$x_{0}<x_{L}<x_{U}$ and negative if $x_{L}<x_{U}<x_{0}$.
\qed%

\subsection{Proof of \texorpdfstring{\Cref*{thm:hetero_identif}}{Proposition~\ref*{thm:hetero_identif}}}
Let $\tau$ be a Rademacher random variable independent of $(D, \Ra, \vec{U})$,
i.e., $P(\tau=1 \mid D, \Ra, \vec{U})=P(\tau=-1 \mid D, \Ra, \vec{U})=1/2$. Since the
distribution of $\Ra$ is symmetric around zero, $\Ra$ has the same distribution as
$|\Ra| \times \tau$, and thus $(X, \vec{U})$ has the same distribution as
$(|X|\tau, \vec{U})$. Let $\tilde{U}$ be uniform on $[0,1]$ independently of
$(D, \Ra)$, and let $\phi_\tau \colon \mathbb{R} \to \mathbb{R}$ and
$\vec{\phi}_{\vec{U}} \colon \mathbb{R} \to \mathbb{R}^{m-1}$ be measurable
functions such that $(\tau, \vec{U})$ has the same distribution as
$(\phi_\tau(\tilde{U}), \vec{\phi}_{\vec{U}}(\tilde{U}))$ (see the discussion
after \cref*{thm:factor} on the construction of such functions). Then it follows
that $(X, \vec{U})$ has the same distribution as
$(|X|\phi_\tau(\tilde{U}), \vec{\phi}_{\vec{U}}(\tilde{U}))$, and the conclusion
of the proposition obtains by defining
$\tilde{\vec{\psi}}(x, \tilde{u}) \equiv \vec{\psi}(|x|\phi_\tau(\tilde{u}),
\vec{\phi}_{\vec{U}}(\tilde{u}))$.\qed%

\subsection{Proof of \texorpdfstring{\Cref*{thm:hetero_iv}}{Proposition~\ref*{thm:hetero_iv}}}
Since $\vec{\gamma}(\vec{U})$ is independent of $(X, D)$ with mean zero,
\begin{align*}
\cov(\vec{Y}, Z \mid X) &= \cov(\vec{\gamma}(\vec{U}), (\Psi_1(X)+\gamma_1(\vec{U}))(D-E[D]) \mid X) \\
&= \cov(\vec{\gamma}(\vec{U}), \gamma_1(\vec{U})) \lbrace E[D\mid X]-E[D]\rbrace.
\end{align*}
The law of total covariance therefore implies
\begin{align*}
\cov(\vec{Y}, Z) &= E[\cov(\vec{Y}, Z \mid X)] + \cov(E[\vec{Y} \mid X], E[Z \mid X]) \\
&= 0 + E[\vec{\Psi}(X) \lbrace E[Z \mid X] - E[Z] \rbrace].
\end{align*}
The result now follows from \Cref*{theorem:rr} and \Cref{lemma:finite_integral_omega}, with weights given by
\begin{align*}
\check{\omega}(x) &\equiv E[\1{X \geq x}\lbrace E[Z \mid X] - E[Z] \rbrace] \\
&= \cov(\1{X \geq x}, E[Z \mid X, D]) \\
&= \cov(\1{X \geq x}, \Psi_1(X)(D-E[D])),
\end{align*}
where the last equality follows from
\begin{equation*}
  \pushQED{\qed}
  E[Z \mid X, D] = E[Y_1 \mid X, D](D-E[D]) = \Psi_1(X)(D-E[D]). \qedhere
  \popQED
\end{equation*}

\subsection{Proof of \texorpdfstring{\Cref*{thm:factor}}{Proposition~\ref*{thm:factor}}}
Let $Q_j(\tau \mid \tilde{Y}_{j-1}, \tilde{Y}_{j-2}, \dotsc, \tilde{Y}_1)$
denote the $\tau$-th quantile of $\tilde{Y}_j$ conditional on
$\tilde{Y}_{j-1}, \tilde{Y}_{j-2}, \dotsc, \tilde{Y}_1$. Now construct an
$n$-dimensional vector $\vec{Y}^*=(Y_1^*, \dotsc, Y_n^*)$ as follows. First set
$Y_1^*\equiv \tilde{Y}_1$. Then for $j > 1$, let
$Y_j^* \equiv Q_j(\bar{U}_{j-1} \mid
\tilde{Y}_{j-1}=Y_{j-1}^*, \dotsc, \tilde{Y}_1=Y_1^*)$. Standard arguments yield
that $\vec{Y}^*$ has the same distribution as $\tilde{\vec{Y}}$. Consequently,
$\bar{\vec{Y}} \equiv \vec{\Upsilon}^{-1}(\vec{Y}^*)$ has the same distribution
as $\vec{Y}=\vec{\Upsilon}^{-1}(\tilde{\vec{Y}})$. The mapping from
$(\tilde{X}, \bar{U}_1,\dotsc, \bar{U}_{n-1})$ to $\vec{Y}^*$ is continuous by the
assumptions on $Q_j$, and so is the implied $\bar{\vec{\psi}}$ mapping by
continuity of $\vec{\Upsilon}^{-1}$.%
\qed%

\subsection{Proof of \texorpdfstring{\Cref*{theorem:rr}}{Lemma~\ref*{theorem:rr}}}

This result follows directly from \Cref*{theorem:rr-covariates} by letting
$\vec{W}$ equal a constant.
\qed%

\subsection{Proof of \texorpdfstring{\Cref*{theorem:rr-covariates}}{Lemma~\ref*{theorem:rr-covariates}}}

Observe
\begin{align*}
    E\left[\int\omega(x, \vec{W})g'(x, \vec{W})\, dx\right]
    &      =          E\left[\int_{I_{\vec{W}}}
      E[\1{X\geq  x\geq x_{0}(\vec{W})}\alpha(X, \vec{W})\mid
      \vec{W}] g'(x,\vec{W})
      \,dx\right]\\
    &\quad      -
      E\left[\int_{I_{\vec{W}}}
      E[\1{X< x<x_{0}(\vec{W})}\alpha(X, \vec{W})\mid
      \vec{W}] g'(x, \vec{W})
      \, dx\right]\\
    &=E\left[\int_{I_{\vec{W}}}
      \1{X\geq x\geq x_{0}(\vec{W})}\alpha(X, \vec{W}) g'(x, \vec{W})
      \, dx\right]\\
    &\qquad -
      E\left[\int_{I_{\vec{W}}}
      \1{X< x< x_{0}(\vec{W})}\alpha(X, \vec{W}) g'(x, \vec{W})
      \, dx\right]\\
    &=E\left[\1{X\geq x_{0}(\vec{W})}
      \alpha(X, \vec{W})
      (g(X, \vec{W})-g(x_{0}(\vec{W}), \vec{W}))
      \right]\\
    &\qquad -
      E\left[\1{X< x_{0}(\vec{W})}
      \alpha(X, \vec{W}) (g(x_{0}(\vec{W}), \vec{W})
      -g(X, \vec{W}))
      \right]\\
    &=E\left[\alpha(X, \vec{W})(g(X, \vec{W})-
      g(x_{0}(\vec{W}), \vec{W}))\right] \\
    &=E\left[\alpha(X, \vec{W})g(X, \vec{W})\right],
  \end{align*}
where the first equality uses the fact that since
$E[{\alpha(X, \vec{W})}\mid \vec{W}]=0$ by condition (iii),
$\omega(x, \vec{w})= -E[\1{X< x}\alpha(X, \vec{w})\mid
\vec{W}=\vec{w}]$, the second equality uses Fubini's theorem, which is
justified since both integrals exist by condition (iv), the third equality
follows by the fundamental theorem of calculus and condition (ii), the fourth equality collects terms, and the last equality uses
iterated expectations, which is justified since
\begin{align*}
  &E\left[\abs{\alpha(X, \vec{W})g(x_{0}(\vec{W}), \vec{W})}\right]
  \\
  &\leq
  E\left[\abs{\alpha(X, \vec{W})g(X, \vec{W})}\right]
  +E\left[\abs{\alpha(X, \vec{W})(g(X, \vec{W})-
    g(x_{0}(\vec{W}), \vec{W}))}\right]\\
  &\leq
  E\left[\abs{\alpha(X, \vec{W})g(X, \vec{W})}\right]
  +E\left[\left|\alpha(X, \vec{W})\int_{x_{0}(\vec{W})}^{X}\abs{g'(x, \vec{W})}\, dx\right|\right]<\infty,
\end{align*}
by conditions (iv) and (v).
\qed%

\subsection{Proof of \texorpdfstring{\Cref*{theorem:ols_covariates}}{Proposition~\ref*{theorem:ols_covariates}}}

Observe that under either condition (a) or condition (b),
\begin{multline*}
  E[(X-\pi(\vec{W}))g(X, \vec{W})]\\
  =
  E[(X-\pi^{*}(\vec{W}))g(X, \vec{W})]+
  E\left[(\pi^{*}(\vec{W})-\pi(\vec{W}))\int \lambda({x}, \vec{W})
  g'(x, \vec{W})\, dx\right].
\end{multline*}
Applying \Cref*{theorem:rr-covariates} with
$\alpha(X, \vec{W})=X-\pi^{*}(\vec{W})$ and
$x_{0}(\vec{W})=\pi^{*}(\vec{W})$ yields
\begin{equation*}
  E[(X-\pi^{*}(\vec{W}))g(X, \vec{W})]
  =E\left[\int \omega^{*}(x, \vec{W})g'(x, \vec{W})\, dx\right].
\end{equation*}
Note that condition (iv) of \Cref*{theorem:rr-covariates} follows from a similar argument as in \Cref{lemma:finite_integral_omega} (conditional on $\vec{W}$).

Since $(X-\pi(\vec{W}))$ is orthogonal to $\pi(\vec{W})$ and to a constant function,
\begin{align*}
  \var(X-\pi(\vec{W}))&=E[(X-\pi^{*}(\vec{W}))X]
  +E[(\pi^{*}(\vec{W})-\pi(\vec{W}))\pi^{*}(\vec{W})]\\
  &=E[(X-\pi^{*}(\vec{W}))X]
  +E\left[(\pi^{*}(\vec{W})-\pi(\vec{W}))\int \lambda(x, \vec{W})\, dx\right],
\end{align*}
and it follows that the weights integrate to one. The last statement of the proposition can be shown using the same argument as in the proof of \Cref*{thm:proxy}.
\qed%

\end{appendices}

\clearpage%
\phantomsection%
\addcontentsline{toc}{section}{References}
\bibliography{ref}